\documentclass[a4paper,UKenglish,cleveref, autoref, thm-restate, numberwithinsect]{lipics-v2021}

\usepackage{stmaryrd,placeins}
\usepackage{amsthm}
\usepackage{tikz,tikz-qtree}
\usetikzlibrary{decorations.pathmorphing}
\usetikzlibrary{backgrounds,calc,positioning}
\usepackage{complexity,bm}
\usepackage{mathtools}
\usepackage{forest}
\usepackage[ruled,linesnumbered]{algorithm2e}

\nolinenumbers

\colorlet{darkgreen}{green!80!black}
\colorlet{lightblue}{blue!40!white}

\newcommand{\alp}{{\mathsf{alph}}}
\newcommand{\depth}{\mathsf{depth}}
\newcommand{\calA}{\mathcal{A}}
\newcommand{\calB}{\mathcal{B}}
\newcommand{\calG}{\mathcal{G}}
\newcommand{\calT}{\mathcal{T}}
\newcommand{\calD}{\mathcal{D}}
\newcommand{\calH}{\mathcal{H}}
\newcommand{\bigO}{\mathcal{O}}
\newcommand{\valA}[1]{\llbracket #1 \rrbracket}
\newcommand{\valGA}[2]{\valA{#2}_{#1}}

\newcommand{\reducecommon}[1]{\operatorname{#1-\mathsf{reduce}}}
\newcommand{\lreduce}{\reducecommon{\ell}}
\newcommand{\rreduce}{\reducecommon{\mathsf{r}}}
\newcommand{\pos}{{\mathsf{p}}}
\newcommand{\Bool}{\mathsf{Bool}}
\newcommand{\encl}{\mathsf{encl}}
\newcommand{\ch}{\mathsf{ch}}

\newcommand{\pat}{\mathsf{path}}

\newcommand{\red}{\mathsf{reduce}}
\newcommand{\Path}{\mathsf{P}_{\mathsf{t}}}
\newcommand{\Ssf}{\mathsf{S}}
\newcommand{\Sa}{\mathsf{S}_{\mathsf{a}}}
\newcommand{\St}{\mathsf{S}_{\mathsf{t}}}
\newcommand{\Sat}{\mathsf{S}_{\mathsf{at}}}

\renewcommand{\succ}{\textsf{succ}}
\renewcommand{\min}{\textsf{min}}
\renewcommand{\max}{\textsf{max}}
\newcommand{\rest}{\mathord\restriction}
\newcommand{\sfR}{\mathsf{R}}
\newcommand{\sfL}{\mathsf{L}}

\hideLIPIcs

\title{Ranked MSO-enumeration over compressed words}

\author{Markus Lohrey}{Universit\"at Siegen, Germany}{lohrey@eti.uni-siegen.de}{https://orcid.org/0000-0002-4680-7198}{}

\authorrunning{M.\ Lohrey}



\begin{document}

\maketitle
\begin{abstract}
It is shown that the ranked query enumeration problem for a fixed MSO-query on strings
 can be solved with linear preprocessing and constant delay in the grammar-compressed setting, where the 
 input string is given by a so-called straight-line program, i.e., a context-free grammar that produces
 exactly one string. Moreover, `ranked' means that the output tuples of the MSO-query are printed in a specific
 order  that has to be MSO-definable. This is the first result for ranked query enumeration on compressed data.
 A corollary of this result is that for a fixed polyregular function $f$ and a word $w$ that is given by a straight-line program
 of size $n$, one can list after preprocessing time $\bigO(n)$ the symbols in $f(w)$ from left to right with constant delay, which
 generalizes a result of Bojanczyk for the case where $w$ is uncompressed.
 The proofs for these results are based on factorization trees, which are made accessible to the grammar-compressed setting
 (a contribution of independent interest).
\end{abstract}

\section{Introduction}

The evaluation of queries formulated in \emph{monadic second order logic} (MSO) is a classical problem in database theory and finite model theory. If we 
allow arbitrary structures, then already the \emph{model-checking problem} for MSO (and even first-order logic) is 
\textsf{PSPACE}-complete. Here, the input consists of the structure and the formula.
In real applications from data base theory or verification, the structure can be huge whereas the formula is often small. This motivates
the study of \emph{data complexity}, where the formula is fixed and not part of the input. Still, MSO-model-checking is hard for all levels
of the polynomial time hierarchy with respect to data complexity. In order to get efficient algorithms, one has to restrict the class of allowed input structures.
A famous result in this context is Courcelle's  meta-theorem~\cite{Courcelle1990}, saying that
the data complexity of MSO-model-checking is linear time on the class of all structures of treewidth at most $k$ (for every fixed $k$).
In this paper, we will only consider data complexity.

The model-checking problem for MSO is the same as the evaluation problem for \emph{boolean} MSO-queries. In data base theory, one is usually 
interested in queries with free variables and the computation of all query answers. 
Here, we consider  MSO-queries $\Phi(x_1, \ldots, x_k)$ with free first-order variables $x_i$.
The goal is then to compute all tuples $(a_1, \ldots, a_k)$ in an input structure $\calA$ (that models the input data) such that $\calA \models \Phi(a_1, \ldots, a_k)$ 
holds.
 The number of these tuples is clearly bounded by $|\calA|^k$ which can be very large in applications, where the input data are large.
In order to have a reasonable notion of efficient algorithms for computing query answers, the concept of \emph{enumeration algorithms} has been introduced.
Such an algorithm starts for a given input structure $\calA$ with a \emph{preprocessing phase} building a suitable data structure for the \emph{enumeration phase} that
starts right after the preprocessing phase. In the enumeration phase the algorithm computes all tuples $(a_1, \ldots, a_k)$ satisfying the (fixed) MSO-query
$\Phi(x_1, \ldots, x_k)$ without producing a tuple twice.
The gold standard for the efficiency of an enumeration algorithm is a \emph{linear preprocessing time} $\bigO(|\calA|)$  and \emph{constant delay}, where
the latter means that there is a constant bounding the following computation times: (i) the time for computing the first output tuple, (ii) the time
between outputting two tuples, and (iii) the time between outputting the last tuple and final termination. Note that constant delay makes only sense 
for data complexity since otherwise the length $k$ of the tuples would not be fixed. Moreover constant delay assumes that every element of the input
structure $\calA$ fits into constantly many output registers.

Note that in the above setting, the MSO-formula only contains free first-order formulas. One may also allow queries
 $\Phi(x_1, \ldots, x_k, X_1, \ldots, X_l)$ with free set variable $X_i$. Although this is not the focus of the paper, we mention 
 it because some of the related papers cited below consider this more general setting. Clearly, the output tuples then no longer fit into 
 constant space and therefore cannot be printed in constant time. The best possible delay bound is then \emph{output-linear delay}, meaning that
 the delay for printing an output tuple is linearly bounded in the size of the tuple. For a fixed query 
$\Phi(x_1, \ldots, x_k)$ containing only free first-order variables, output-linear delay is equivalent to constant delay.

Enumeration algorithms with linear preprocessing and constant delay or output-linear delay are known for several classes
of queries and structures, including  MSO-queries on strings, trees and structures with bounded treewidth~\cite{AmarilliBMN19,Bagan06,Courcelle2009,KazanaS13,MMN22,Niewerth18,NiewerthSegoufin2018} and for regular document spanners (which is a subclass of MSO-queries) on strings~\cite{AmarilliEtAl2021,FlorenzanoEtAl2020}. 

As mentioned above, the input structure for query evaluation can be huge in applications. This has motivated the investigation of query enumeration algorithms for structures that
are given in a \emph{compressed} form. The compression format that has been studied so far in the context of query enumeration is \emph{grammar-compression}. For strings
this means that the input string is given by a context-free grammar that produces exactly one string.
Such grammars are also known as \emph{straight-line algorithm} (SLP for short) and have received a lot of attention in data compression, information theory, stringology and even areas
like computational topology and group theory; see the survey~\cite{Loh12survey} for more details. In data base theory, 
query enumeration algorithms on compressed structures have been first studied in the context of strings and document spanners 
\cite{SchmidS21,SchmidSchweikardt2022,SchmidS22}. This work has been extended to full MSO in \cite{MunozRiveros2025}, where it is shown
that for a fixed MSO-query and a given input word that is represented by an SLP $\calG$, the set of all query answers can be enumerated 
 with linear processing (meaning time $\bigO(|\calG|)$
where $|\calG|$ is the size of the SLP $\calG$) and output-linear delay 
(the authors of \cite{MunozRiveros2025} consider MSO-formulas with free set variables). In \cite{LohreyS26} this result has been extended to unranked
trees that are given by so-called forest straight-line programs \cite{GasconLMRS20} 
(an extension of SLPs that allows to compress unranked trees), whereas \cite{LohreyMS25} goes one step
further to grammar-compressed graphs of bounded degree but restricts to queries expressed in first-order logic.

A drawback in all the above mentioned results on query answer enumeration is the fact that the output tuples are 
printed in some opaque order that cannot be controlled by the user. To overcome this drawback, researchers have developed the concept
of \emph{ranked} enumeration, where the output tuples are printed in a specific order that can be specified in a suitable formalism.
Ranked enumeration for MSO-queries has been considered in \cite{BourhisGJR21,GMS24} for strings and \cite{AmarilliBCM24}
for trees. In \cite{BourhisGJR21,GMS24}, query answers are enumerated in the order of decreasing weights, where the weight of an input word
is an element of an ordered abelian group that is computed by a so-called cost transducer.
In contrast, \cite{AmarilliBCM24} uses ranking functions satisfying a subset-monotonicity property.

In the worst case, the delay of the enumeration algorithms from \cite{AmarilliBCM24,BourhisGJR21,GMS24} is logarithmic in the size of the input structure
(with some improvements in \cite{GMS24} over \cite{BourhisGJR21}). Moreover,
the input structure is not compressed  in \cite{AmarilliBCM24,BourhisGJR21,GMS24}.
Ranked enumeration on compressed data has not been studied so far. In this paper, we make the first step in this direction. 
We consider a fixed MSO-formula with $k$ free first-order variables, an SLP-compressed input string and a linear order on $k$-tuples of string positions that is MSO-definable 
(a good example for such an order is the lexicographic order on $k$-tuples of positions; see Example~\ref{ex-lex}). 
Our main result (Theorem~\ref{thm main}) states that in this setting, the set of all query answers can be enumerated in the specified order with linear preprocessing and constant 
delay.  

Before we say something about the proof ingredients, 
let us mention an easy corollary concerning \emph{polyregular functions}. 
These string-to-string functions can be characterized in many different ways (e.g. by MSO-interpretations, pebble-two-way transducers,
certain imperative as well as functional programming languages, and as the class obtained from certain primitive functions by closing under composition)
 and have received a lot of attention in recent years; see \cite{Boj18,BojanczykKL19} for more details.
In \cite{Boj18}, it is shown that for a fixed polyregular function $f$ and a given input word $w$, one can enumerate the symbols in $f(w)$  
from left to right after linear time preprocessing in constant delay. Note that $|f(w)|$ is bounded polynomially in $|w|$, where the polynomial depends on $f$.
Recently, it has been shown that polyregular functions do not preserve compression, meaning that there exists a specific polyregular function $f$, 
a family of strings $w_n$
such that every $w_n$ has an SLP of size $\bigO(n)$, whereas the size of a smallest SLP for $f(n)$ is lower bounded by $2^n$; see \cite[Lemma 13]{BL26} and its proof.
Our corollary of Theorem~\ref{thm main} extends the above mentioned enumeration result from \cite{Boj18} to the compressed setting: For every fixed
polyregular function $f$ and a given input SLP $\calG$ producing the string $w$, one can enumerate the symbols in $f(w)$  
from left to right after a linear time preprocessing in constant delay. 

We conclude the introduction with a few words on the proof ingredients for our main result. A main tool are factorization trees for words, which have
been applied in database theory before  \cite{BojanczykP08}. A factorization tree for a word $w \in \Sigma^+$ with respect to a morphism $h : \Sigma^* \to M$ into a finite monoid
is a constant-height parse tree of $w$ such that for every node with more than two children, all children evaluate in $M$ to the same idempotent element of $M$.
The existence of such  factorization trees for every word has been shown by Simon \cite{Simon90}.
In our setting, the word $w$ is given by an SLP $\calG$. We introduce the notion of a \emph{Simon SLP}, which, roughly speaking, is an SLP that is compatible
with the structure of a factorization tree for the string produced by the SLP. We show in Section~\ref{sec SSLP} that every SLP can be transformed in linear time into
an equivalent Simon SLP.
We then combine these Simon SLPs with a technique for constant-time two-way traversal
of SLP-compressed strings from \cite{LohreyMR18} (this technique has to be suitably extended for our purpose; see Section~\ref{sec-traversal}) in order to extend the above mentioned
algorithm from \cite{Boj18} for enumerating the symbols in the image $f(w)$ of a fixed polyregular function $f$ to the setting, where $w$ is given by an SLP.

\subparagraph{Related work.}
We mentioned already the existing work on ranked MSO-enumeration.
Related to our work is also the following recent result of Muñoz \cite{Munoz26} on direct ranked access: 
For a fixed MSO-query (containing only free first-order variables) and a given SLP of size $n$
producing a  string of length $N$, one can compute in time $\bigO(n)$ a data structure that allows to compute for a given number $m$ in time 
$\bigO(\log N)$ the lexicographically $m$-th query answer.

\section{Preliminaries}

For a mapping $f : M \to N$ and a subset $P \subseteq M$ we write $f\rest_P : P \to N$ for the restriction of 
$f$ to $P$. For integers $i \le j$ we use the notation $[i,j] = \{ k \in \mathbb{Z} : i \le k \le j \}$.

Given an alphabet of symbols $\Sigma$, $\Sigma^*$ denotes the set of all finite words
over the alphabet $\Sigma$, including the empty word $\varepsilon$. The set of non-empty
words is denoted by $\Sigma^+ = \Sigma^* \setminus \{\varepsilon\}$.
The length of a word $s$ is denoted with $|s|$. If $s = a_1 a_2 \cdots a_n$ ($a_i \in \Sigma$) then 
we define $s[i] = a_i$ for $i \in [1,n]$ and $s[i,j] = a_i a_{i+1} \cdots a_j$ for $1 \le i \le j \le n$.
The alphabet of $s$ is $\alp(s) = \{a_1, a_2, \ldots, a_n\}$.

\subparagraph*{Finite automata and transducers.}
We assume some knowledge from  automata theory; see e.g. \cite{Saka09} for further details. 
Consider a nondeterministic finite automaton $\calA$ over the alphabet $\Sigma$ with $n$ states.
We say that $\calA$ is a trim NFA if  every state of $\calA$ lies on a path from the initial state to a final state. Clearly, every NFA can be transformed
into an equivalent trim NFA without increasing the number of states.

One can represent the dynamics of an NFA $\calA$ by a monoid of 
boolean $(n \times n)$-matrices. For this, assume w.l.o.g.~that the set of states of $\calA$ is
$\{1,\ldots,n\}$. For every letter $a \in \Sigma$ define the boolean matrix $M_a$ by setting the entry of $M_a$ in row $i$ and column $j$ to $1$ if there is a transition
from state $i$ to state $j$ reading the letter $a$, otherwise the entry is $0$. Then $\Bool(\calA)$ is the monoid generated by all matrices $M_a$ ($a \in \Sigma$)
under multiplication of boolean matrices (where scalar addition and multiplication are the boolean operators $\vee$ and $\wedge$, respectively).
For a subset $\Gamma \subseteq \Sigma$ we write 
$\Bool(\calA, \Gamma)$ for the submonoid of $\Bool(\calA)$ 
generated by all matrices $M_a$ with $a \in \Gamma$.
If $h : \Gamma^* \to \Bool(\calA, \Gamma)$ is the morphism defined by $a \mapsto M_a$ for $a \in \Gamma$ and $w \in \Gamma^*$, then the entry 
in row $i$ and column $j$ of the matrix $h(w)$ is $1$ if and only if there is a path in $\calA$ from state $i$ to state $j$ that is labelled with the word $w$.

A nondeterministic transducer $\calT$ is a nondeterministic automaton with output that reads in each 
transition  exactly one input symbol $a$, outputs a word $u$ over the output alphabet and changes the state from $p$ to $q$.
This transition is also written as $(p, a/u, q)$. 
Since $\calT$ is nondeterministic, a word $s$ is translated into a finite set of words $\calT(s)$.

\subparagraph*{Trees.}
In the following we consider \emph{finite rooted ordered trees}. Here, ordered means that the children of a node $u$ are linearly ordered (when drawing a tree
this is the usual left-to-right order). We denote this linear order by $<_u$.
 These linear orders $<_u$ induce the \emph{depth-first left-to-right order} $<_T$ on the nodes of $T$: For nodes $v_1 \neq v_2$ of $T$ we have
 $v_1 <_T v_2$ if either $v_2$ is a descendant of $v_1$ or there is a node $u$ with two children $u_1 <_u u_2$ such that $v_i$ is a 
 descendant of $u_i$ for $i \in \{1,2\}$. If $v_2$ is a descendant of $v_1$ in the tree $T$ then we write $[v_1, v_2]$ for the unique path from $v_1$ down to $v_2$.
Often, we identify $[v_1, v_2]$ with the set of nodes on this path.

\subsection{Factorization trees}
Throughout this section, we fix a finite alphabet $\Sigma$, a finite monoid $M$, and a  monoid homomorphism $h : \Sigma^* \to M$.
By $E(M) = \{ e \in M : e^2 = e\}$ we denote the set of idempotents of $M$. 
A \emph{factorization tree} for a word $s = a_1 a_2 \cdots a_n$ ($a_i \in \Sigma$) with respect to $h$ is
a finite rooted ordered tree $T$ such that the following holds:
\begin{itemize}
\item $T$ has $n$ leaves $v_1 <_T v_2 <_T \cdots <_T v_n$, where $v_i$ is labelled with $a_i$.
This allows to assign to every node $u$ of $T$ a word $w(u) \in \Sigma^*$: $w(v_i) = a_i$ and if $u$ is an internal node with children $u_1 <_u u_2 <_u \cdots <_u u_k$ then $w(u) = w(u_1) w(u_2) \cdots w(u_k)$ (hence, $w(r) = s$ if $r$ is the root of $T$).
The monoid element $h(u) \in M$ associated to node $u$ is $h(u) = h(w(u))$.
\item  Every node $u$ is either a leaf or has at least two children.
\item If $u$ has $k \geq 3$ children $u_1, u_2, \ldots, u_k$ then there is an $e \in E(M)$ such that
$h(u_i) = e$ for all $1 \le i \le k$ (and hence also $h(u) = e$). In this case we call $u$ an \emph{idempotent node}.
Nodes with exactly two children are called \emph{binary nodes}.
\end{itemize}
Figure~\ref{fact tree} shows a factorization for the word $abccabdabcabe$, where the $i$-th leaf is labelled with the $i$-th letter.
We assume that $h(ab) = h(c) = h(e)$ is idempotent.

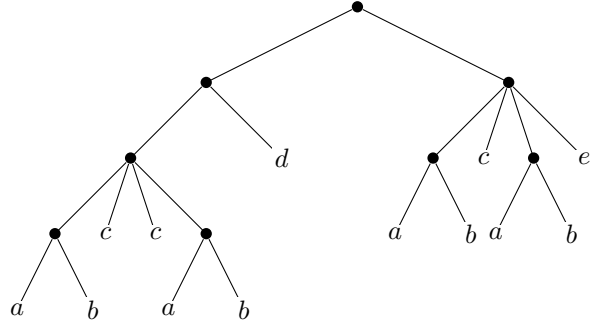
\begin{figure}[t]
 \tikzset{alpha/.style={inner sep = 1pt, fill=white}}
 \tikzset{round/.style={inner sep = 1.5pt, circle, fill=black}}
  \tikzset{empty/.style={inner sep = 0pt}}
\centering{
\begin{tikzpicture}
\draw (0,0) node[round] (0) {} ;
\draw (0) -- ++(-2,-1)  node[round]  (1) {} 
          (0) -- ++(2,-1)  node[round]  (2) {}
          (1) -- ++(-1,-1)  node[round]  (3) {}
          (1) -- ++(1,-1)  node[alpha]  {$d$}
          (3) -- ++(-1,-1)  node[round]  (4) {}
          (3) -- ++(-.33,-1)  node[alpha]  {$c$}  
          (3) -- ++(.33,-1)  node[alpha]  {$c$} 
          (3) -- ++(1, -1)  node[round] (5) {} 
          (4) -- ++(-.5,-1)  node[alpha]  {$a$}
          (4) -- ++(.5,-1)  node[alpha]  {$b$} 
          (5) -- ++(-.5,-1)  node[alpha]  {$a$}
          (5) -- ++(.5,-1)  node[alpha]  {$b$} 
          (2) -- ++(-1,-1)  node[round]  (6) {}
          (2) -- ++(-.33,-1)  node[alpha]  {$c$}  
          (2) -- ++(.33, -1)  node[round] (7) {} 
          (2) -- ++(1,-1)  node[alpha]  {$e$} 
          (6) -- ++(-.5,-1)  node[alpha]  {$a$}
          (6) -- ++(.5,-1)  node[alpha]  {$b$} 
          (7) -- ++(-.5,-1)  node[alpha]  {$a$}
          (7) -- ++(.5,-1)  node[alpha]  {$b$} ;
\end{tikzpicture}}
\caption{A factorization tree.}  \label{fact tree}
\end{figure}

Simon \cite{Simon90} showed that every word   $s \in \Sigma^+$ has a factorization tree of height at most $9 |M|$.
The bound in the following theorem is due to Kufleitner \cite{Kuf08} and is sharp.
\begin{theorem}
Every word $s \in \Sigma^+$ has a factorization tree of height at most $3 |M|$.
\end{theorem}
A factorization tree $T$ can (as any rooted ordered tree) be written as an expression $\gamma(T)$ with brackets $($ and $)$.
Formally, if $T$ consists of a single leaf that is labelled with $a \in \Sigma$ then $\gamma(T) = (a)$. Otherwise, if the root 
of $T$ has $k \geq 2$ children, let $T_i$ ($1 \le i \le k$) be the subtree of $T$ rooted in the $i$-th child of the root. Then
we have $\gamma(T) = ( \gamma(T_1) \gamma(T_2) \cdots \gamma(T_n))$.
Every subtree $T'$ of $T$ can identified with an occurrence of $\gamma(T')$ in $\gamma(T)$ in the natural way.
Since the nesting depth of brackets is bounded by the  constant $3|M|$, the set 
\[ \{ \gamma(T) : \mbox{$T$ is a factorization tree w.r.t.~$h$ of height at most $3|M|$ for some word $s \in \Sigma^+$} \}\]
is a regular language over the alphabet $\Sigma \cup \{(,)\}$. This implies the following result that is 
stated in \cite[Lemma~3]{Boj09}. 
\begin{theorem} \label{thm nondet transducer}
There is a nondeterministic 
transducer $\calT_h$  (that only depends on the monoid $M$ and the morphism $h : \Sigma^* \to M$)
such that for every word $s \in \Sigma^*$ we have $\calT_h(s) = \{ \gamma(T) : \text{$T$ is a factorization tree for $s$ w.r.t.~$h$ of height at most
$3 |M|$} \}$. 
\end{theorem}

\subsection{Monadic second-order logic over words} \label{sec-MSO}
 
 We assume that the reader is familiar with \emph{monadic second-order} (MSO) logic over words; see \cite{Str94} for more details.
 Here, a non-empty word $w = a_1 a_2 \cdots a_n \in \Sigma^+$ of length $n$ is identified with a finite structure (a so-called \emph{word structure}) consisting of the universe
 $[1,n]$ (all positions in the word $w$), the binary relation $<$ (the standard order on integers) and the unary relations $P_a = \{ i : a_i = a \}$ for $a \in \Sigma$.
 We will consider MSO-formulas over words with free first-order variables that range over positions in words. For an MSO-formula $\Phi$ and first-order variables
 $x_1, \ldots, x_k$ we write $\Phi(x_1,\ldots,x_k)$ to express that $x_1, \ldots, x_k$ are the free variables of $\Phi$.
 For a word $w \in \Sigma^+$ of length $n$, an MSO-formula  $\Phi(x_1,\ldots,x_k)$ and positions $p_1, \ldots p_k \in [1,n]$ we write
 $w \models \Phi(p_1, \ldots, p_k)$ to express that $\Phi$ is true in the word $w$ when every variable $x_i$ takes the value $p_i$.
 Moreover, we define $\valA{\Phi}_w := \{ (p_1, \ldots p_k) \in [1,n]^k : w \models \Phi(p_1, \ldots, p_k) \}$.
 
 \subsection{Polyregular functions and MSO-interpretations} \label{sec-polyreg}
 
 As already mentioned in the introduction, the class of polyregular string-to-string functions has several alternative definitions. For us, the important one
 is the one based on MSO-string-to-string interpretations \cite{BojanczykKL19}. For this, consider two finite alphabets $\Sigma$ and $\Gamma$.
 An MSO-string-to-string interpretation is a function $f : \Sigma^+ \to \Gamma^+$ for which there exist
 a $k \geq 1$ and MSO-formulas (all over words from $\Sigma^+$) $\Phi(x_1,\ldots,x_k)$, $\Phi_<(x_1,\ldots,x_k, y_1,\ldots,y_k)$, and $\Phi_b(x_1,\ldots,x_k)$ for all $b \in \Gamma$
 such that the following hold for every word $w \in \Sigma^+$ of length $n$  (we write $\bar{x}$ for $(x_1, \ldots,x_k)$ and similarly for $\bar{y}$):
 \begin{itemize}
 \item  $w \models \forall \bar{x}, \bar{y} : \Phi_<(\bar{x}, \bar{y}) \to (\Phi(\bar{x}) \wedge \Phi(\bar{y}))$
 and  $w \models \forall \bar{x} : \Phi_b(\bar{x}) \to \Phi(\bar{x})$ for every $b \in \Gamma$.
\item The structure with universe $\valA{\Phi}_w \subseteq [1,n]^k$,  the binary 
relation $\valA{\Phi_<}_w \subseteq \valA{\Phi}_w \times \valA{\Phi}_w$ and the unary relations $\valA{\Phi_b}_w \subseteq \valA{\Phi}_w$ for $b \in \Gamma$
is a word structure over the alphabet $\Gamma$ in the sense of Section~\ref{sec-MSO} and the corresponding word is $f(w)$.
 \end{itemize}
The polyregular functions are exactly the MSO-string-to-string interpretations \cite{BojanczykKL19}.
 
 \begin{example} \label{ex-lex}
 A good example for an MSO-definable linear on tuples from $[1,n]^k$ ($k \geq 1$, $n \geq 1$)
 is the lexicographic order $<_{k,n}^{\textsf{lex}}$. 
For every $k \geq 1$ there is an MSO-formula $\Phi_k$ such that for every word $w$ of length $n$, $\valA{\Phi_k}_w$ 
is $<_{k,n}^{\textsf{lex}}$. It can be inductively defined by 
$\Phi_1(x_1, y_1) = (x_1 < y_1)$ and 
$\Phi_k(x_1, \ldots, x_k, y_1, \ldots, y_k) = (x_1 < y_1 \vee (x_1 = y_1 \wedge \Phi_{k-1}(x_2, \ldots, x_k, y_2, \ldots, y_k))$
for $k \geq 2$.
\end{example}
 
\subsection{Straight-line programs} \label{sec-slp}

Let $\Sigma$ be a finite alphabet of terminal symbols.
A \emph{straight-line program} (SLP for short) over the terminal alphabet
$\Sigma$ is a context-free grammar that produces exactly one word. This can
be syntactically enforced by two properties:
\begin{itemize}
\item The context-free grammar is acyclic.
\item For every variable $A$ there is a unique production of the form $A \to w$.
\end{itemize}
Formally, we define an SLP over the terminal alphabet
$\Sigma$ as a pair $\calG = (V,\rho)$, where $V$ is a finite set of variables with $V \cap \Sigma = \emptyset$ and $\rho \colon V \to (\Sigma \cup V)^*$ 
(the \emph{right-hand side mapping}) has the property that the binary relation
$\{ (A,B) \in V \times V : B \in \alp(\rho(A)) \}$ is acyclic.
This allows to define a homomorphism $\rho^* : (V \cup \Sigma)^* \to \Sigma^*$
uniquely by setting $\rho^*(a)=a$ for $a\in \Sigma$ and $\rho^*(A) = \rho^*(\rho(A))$. 
For $A \in V$ we also write $\valGA{\calG}{A}$ for $\rho^*(A)$.
Often, an SLP has a distinguished root (or start) variable $S$. The resulting triple $\calG = (V,\rho,S)$ is then 
also called a \emph{rooted SLP} (rSLP for short) and we define $\valA{\calG} = \valGA{\calG}{S}$. 
An rSLP $\calG$ can be seen as a context-free grammar that produces 
the single string $\valA{\calG}$. 
Our convention is that capital roman letters (typically $A,B,C, S, X, Y$) denote variables from $V$, lowercase roman letters $a,b,c$ denote
terminal symbols from $\Sigma$, and greek letters $\alpha, \beta, \gamma$ denote symbols from $V \cup \Sigma$.
We define the size $|\calG|$ of the SLP $\calG = (V,\rho)$ as $\sum_{A \in V} |\rho(A)|$.

For our purpose it is convenient to require that $|\rho(A)|=2$ for all variables $A$. 
This property only excludes terminal words of length at most one
(which is not a big loss since our considerations are only interesting for long words) and can
be enforced in linear time:

\begin{lemma} \label{lemma-normalization}
A given rSLP $\calG$ in the general sense (where right-hand sides may have arbitrary length) 
such that $\valA{\calG}$ has length at least two can be transformed
in time $\bigO(|\calG|)$ into an rSLP $\calH$ such that
all right-hand sides of $\calH$ have length two and $\valA{\calH} = \valA{\calG}$. 
\end{lemma}

\begin{proof}
The proof is straightforward (see \cite[Proposition~3.8]{Loh14} for a similar statement): One first eliminates all variables $B$ with $\rho(B) = \varepsilon$
by deleting every occurrence of $B$ in a right-hand side $\rho(A)$.
Only $|\calG|$ many such deletions can be done in total.
Next one eliminates variables $A$ with $\rho(A) \in V \cup \Sigma$. For this one computes in linear time the mapping $\mathsf{chain} : V \to V \cup \Sigma$ that
is defined by $\mathsf{chain}(A) = A$ if $|\rho(A)| \ge 2$, $\mathsf{chain}(A) = \mathsf{chain}(B)$ if $\rho(A)=B \in V$ and $\mathsf{chain}(A) = a$ if $\rho(A)=a \in \Sigma$.
Note that for every variable $A$, $\mathsf{chain}(A)$ is either a terminal symbol or a variable $B$ with $|\rho(B)| \geq 2$.
One then keeps only those variables $A \in V$ with $|\rho(A)| \geq 2$.
If the right-hand side of such a variable contains a variable $B$ with $\rho(B) \in V \cup \Sigma$ then we replace
this $B$ by $\mathsf{chain}(B)$. 
Finally, one eliminates right-hand sides of length at least three in the same way as it is done in the construction of the Chomsky normal form.
\end{proof}
In the following we assume that all right-hand sides of SLPs have length two
without mentioning this assumption explicitly. Let $\calG = (V,\rho)$ be such an SLP over the terminal alphabet $\Sigma$. It will be convenient 
to view $\calG$ as a dag (directed acyclic graph) $\calD(\calG) = (V \cup \Sigma,E(\calG))$ where 
the set of edges is
\[
E(\calG) = \{ (A, \ell, \alpha), (A, r, \beta) \in V \times (V\cup\Sigma) : A \in V, \rho(A) = \alpha\beta \}.
\]  
Here, $\ell$ and $r$ stand for `left' and `right', respectively.
Note that $|\calG|$ is the number of edges of the dag $\calD(\calG)$.
\emph{Paths} in $\calD(\calG)$ are edge sequences 
\[ \pi = (A_1, d_1, A_2) (A_2, d_2, A_3) \cdots (A_{k-1}, d_{k-1}, A_k)(A_k, d_k, \alpha) \in E(\calG)^+.
\]
This path is uniquely specified by $A_1$ and the sequence of directions $d_1 d_2 \cdots d_k$. Later, we generalize paths by contracting subpaths that always
move left, respectively, right. Then, it will be important to keep also the variables $A_i$ in the path description.
The length of the above path $\pi$ is $k$ (it is at least one since $\pi \in E(\calG)^+$) and $\pi$ starts in $A_1 \in V$ and ends in $\alpha \in V \cup \Sigma$.
If $\alpha \in \Sigma$ then $\pi$ is a \emph{terminal path} and 
we denote with $\Path(\calG)$ the set of all terminal paths in $\calD(\calG)$ and with 
$\Path(\calG, A) \subseteq \Path(\calG)$ the set of all terminal paths in $\calD(\calG)$ that start in $A \in V$.
If $\calG$ is clear from the context, we write $\Path(A)$ instead of $\Path(\calG, A)$.
The \emph{depth} of $\calG$, $\depth(\calG)$ for short, 
is the maximal length of a path in  $\calD(\calG)$. 

Consider an SLP $\calG = (V,\rho)$ and a variable $A \in V$ producing the word  $w = \valGA{\calG}{A}$ of length $N$.
The directions $d_i$ in paths induce a \emph{lexicographic order} on the set of paths $\Path(A)$, where $\ell < r$.
Then, the function $\pos_{A} : \Path(A) \to [1,N]$ that maps the lexicographically $i$-th path in $\Path(A)$
to position $i$ is bijective. 

For a subset $V' \subseteq V$ of variables we define the subSLP $\calG\rest_{V'}$ as the 
SLP $(V', \rho\rest_{V'})$. The terminal alphabet of $\calG\rest_{V'}$ is 
$\{ \alpha \in V \cup \Sigma : \alpha \in \alp(\rho(A)) \text{ for some } A \in V' \} \setminus V'$.

\begin{example} \label{ex slp}
Consider the SLP $\calG = (\{S,A,B,C,D\}, \rho)$ over the terminal alphabet $ \{a,b\}$, where $\rho$ is given by
$$
\rho(S)= AB, \  \rho(A)= BC, \ \rho(B)= CC, \ \rho(C) = aD, \ \rho(D) = ab .
$$
We have $\valGA{\calG}{S} = (aab)^5$, the size of the SLP is 10 and its depth is 5. 

The dag $\calD(\calG)$ is shown in Figure~\ref{fig:dt}.
For the path 
\[ \pi = (S,\ell, A)(A,\ell,B)(B,r,C)(C,\ell,a) \in \Path(S) \]
we have $\pos_S(\pi) = 4$.
Its lexicographic predecessor in $\Path(S)$ is 
\[ \pi' = (S,\ell, A)(A,\ell,B)(B,\ell,C)(C,r,D)(D,r,b),
\] 
whereas its 
lexicographic successor in $\Path(S)$ is 
\[ \pi'' = (S,\ell, A)(A,\ell,B)(B,r,C)(C,r,D)(D,\ell,a).
\]
In particular, we have $\pos_S(\pi') = 3$ and $\pos_S(\pi'') = 5$.

The terminal alphabet of the SLP $\calG\rest_{\{S,A,B\}}$ contains only $C$ and it produces $C^5$.
\end{example}
 
 \begin{figure}
\tikzset{alpha/.style={inner sep = 1pt, fill=white}}
\tikzset{round/.style={inner sep = 1.2pt, circle, fill=black}}
\centering{
\begin{tikzpicture}
\node[alpha] (S) {$S$};
\node[alpha, below = .6cm of S] (A) {$A$};
\node[alpha, below = .6cm of A] (B) {$B$};
\node[alpha, below = .6cm of B] (C) {$C$};
\node[alpha, below = .6cm of C] (D) {$D$};
\node[alpha, below left = .6cm and .2cm of D] (a) {$a$};
\node[alpha, below right = .6cm and .2cm of D] (b) {$b$};

\draw[->] (S) to[bend right=30] node[left = -.6mm]{$\scriptstyle{\ell}$} (A);
\draw[->] (S) to[bend left=60] node[right = -.6mm]{$\scriptstyle{r}$} (B);
\draw[->] (A) to[bend right=30] node[left = -.6mm]{$\scriptstyle{\ell}$}(B);
\draw[->] (A) to[bend left=60] node[right = -.6mm]{$\scriptstyle{r}$} (C);
\draw[->] (B) to[bend right=30] node[left = -.6mm]{$\scriptstyle{\ell}$} (C);
\draw[->] (B) to[bend left=30] node[right = -.6mm]{$\scriptstyle{r}$} (C);
\draw[->] (C) to[bend right=30] node[left = -.6mm]{$\scriptstyle{\ell}$} (a);
\draw[->] (C) to[bend left=30] node[right = -.6mm]{$\scriptstyle{r}$} (D);
\draw[->] (D) to[bend right=0] node[left = -.6mm]{$\scriptstyle{\ell}$} (a);
\draw[->] (D) to[bend left=0] node[right = -.6mm]{$\scriptstyle{r}$} (b);
\end{tikzpicture}}
\caption{The dag $\calD(\calG)$ for the SLP from Example~\ref{ex slp}.}\label{fig:dt}
\end{figure}
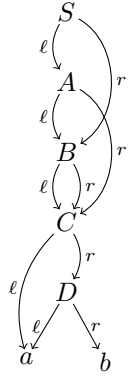

 \subsection{Enumeration algorithms}
 
 We use the standard RAM model for algorithms.
 We consider enumeration algorithms where the output values are enumerated in a specific order.
 This can be formalized as follows. 
 A \emph{ranked enumeration problem} is a function $E$ that maps an input $I$ to a finite word $E(I) \in \Omega^*$
over an alphabet $\Omega$ of output values. In general, the alphabet $\Omega$ may depend on the input $I$. An important
restriction for us is that every element $a \in\Omega$ fits into a constant number of RAM registers, where the bit length of 
RAM registers depends on the input length $|I|$ in a certain way (more about this later).
An enumeration algorithm $A$ for $E$ is an algorithm that computes on input $I$ the word $E(I)$ from left to right. More precisely,
if $E(I) = a_1 a_2 \cdots a_m$ then the algorithm only starts with
  the computation of $a_{i+1}$, once it finishes outputting $a_i$. After the computation of $a_m$ the algorithm outputs the 
 special end symbol $a_{m+1} = \mathsf{end}$.

 The \emph{preprocessing time} of $A$ on input $I$ is the time when the algorithm starts with outputting  $a_1$.
The preprocessing time  of $A$ is the function that maps an $n \in \mathbb{N}$ to the 
 maximum preprocessing time of $A$ over all possible inputs $I$ of length at most $n$.
The \emph{delay} of  $A$ on input $I$ is the  maximal time between the computation of two consecutive output values $a_i$ and $a_{i+1}$,
where $E(I) = a_1 a_2 \cdots a_m$, $a_{m+1} = \mathsf{end}$ and $i \in [1,m]$. 
The algorithm $A$ works in \emph{constant delay} if there is a constant $d$ such that for every input $I$ the delay is bounded by $d$.
The gold standard in the area of enumeration algorithms is linear preprocessing (i.e., the preprocessing time is $\bigO(|I|)$) and 
constant delay. 

We consider enumeration algorithms, where the input $I$ is an rSLP $\calG$. Assume that $N$ is the length of $\valA{\calG}$.
We assume that positions from the interval $[1,N]$ fit into single RAM registers. This is a standard assumption in the 
area of algorithmics on grammar-compressed objects; see e.g. \cite{BLRSSW15,GanardiJL21,MunozRiveros2025}.

\section{SLP-compressed factorization trees} \label{sec SSLP}

Let $M$ be a finite monoid, $h : \Sigma^* \to M$ a homomorphism
and $\calG = (V,\rho)$ an SLP over the  terminal alphabet $\Sigma$. 
Constants that are hidden in the $\bigO$-notation below depend only on the monoid $M$ and the morphism $h$ in the following.
We extend $h$ to $h : (\Sigma \cup V)^* \to M$
by defining $h(A) = h(\rho^*(A))$ for $A \in V$.
For an idempotent $e \in M$ we define the \emph{$e$-part} $\calG_e := \calG\rest_{V_e}$, where
\[ V_e = \{ A \in V : h(A) = e, \rho(A) = \alpha\beta \text{ with } h(\alpha) = h(\beta) = e \}.\]
Let $\Sigma_e$ be the set of terminal symbols of $\calG_e$; it may contain variables and terminals of $\calG$.
Note that $h(\alpha) = e$ for every $\alpha \in V_e \cup \Sigma_e$. Hence, 
the sets $V_e \cup \Sigma_e$ and $V_{e'} \cup \Sigma_{e'}$ are disjoint for idempotents $e \neq e'$.
We define the \emph{idempotent-contracted depth} of $\calG$
as follows: Consider the dag $\calD(\calG)$ and a path $\pi = (A_1, d_1, A_2) \cdots (A_{k-1}, d_{k-1}, A_k)(A_k, d_k, a) \in \Path(A_1)$
 in $\calD(\calG)$
with $A_i \in V$ ($i \in [1,k]$) and $a \in\Sigma$. To simplify notation we write $A_{k+1}$ for $a$.
The  \emph{idempotent-contracted length} of $\pi$ is the length of 
the path obtained from $\pi$ by contracting all maximal subpaths of $\pi$ that are paths in
some $\calD(\calG_e)$ into single edges. It can be also defined as the number of $i \in [1,k]$ such that either $A_i \notin \bigcup_{e \in E(M)} V_e$ 
or $A_i \in V_e$ and $A_{i+1} \in \Sigma_e$ for some $e \in E(M)$. 
 Figure~\ref{fig:id-contract} shows a path of length 9 with idempotent-contracted length 7.
  The idempotent-contracted depth of $\calG$ is the maximal idempotent-contracted length of a path in $\calD(\calG)$.

 \begin{figure}
\tikzset{alpha/.style={inner sep = 1pt, fill=white}}
\tikzset{round/.style={inner sep = 1.2pt, circle, fill=black}}
\centering{
\begin{tikzpicture}
\node[alpha] (1) {$A_1$};
\node[alpha, right = .6cm of 1] (2) {$A_2$};
\node[alpha, right = .6cm of 2,label = above:$V_e$] (3) {$A_3$};
\node[alpha, right = .6cm of 3,label = above:$V_e$] (4) {$A_4$};
\node[alpha, right = .6cm of 4,label = above:$\Sigma_e$] (5) {$A_5$};
\node[alpha, right = .6cm of 5] (6) {$A_6$};
\node[alpha, right = .6cm of 6,label = above:$V_i$] (7) {$A_7$};
\node[alpha, right = .6cm of 7,label = above:$V_i$] (8) {$A_8$};
\node[alpha, right = .6cm of 8,label = above:$\Sigma_i$] (9) {$A_9$};
\node[alpha, right = .6cm of 9] (10) {$a$};

\draw[->,red] (1) to node[above = -.6mm]{$\scriptstyle{\ell}$} (2);
\draw[->,red] (2) to node[above = -.6mm]{$\scriptstyle{\ell}$} (3);
\draw[->] (3) to node[above = -.6mm]{$\scriptstyle{r}$} (4);
\draw[->,red] (3) to[bend right=40] (5);
\draw[->] (4) to node[above = -.6mm]{$\scriptstyle{\ell}$} (5);
\draw[->,red] (5) to node[above = -.6mm]{$\scriptstyle{r}$} (6);
\draw[->,red] (6) to node[above = -.6mm]{$\scriptstyle{r}$} (7);
\draw[->] (7) to node[above = -.6mm]{$\scriptstyle{r}$} (8);
\draw[->,red] (7) to[bend right=40] (9);
\draw[->] (8) to node[above = -.6mm]{$\scriptstyle{\ell}$} (9);
\draw[->,red] (9) to node[above = -.6mm]{$\scriptstyle{\ell}$} (10);
\end{tikzpicture}}
\caption{The idempotent-contracted length of a path.}\label{fig:id-contract}
\end{figure}
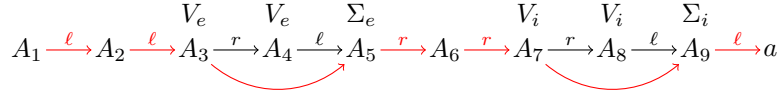

A \emph{Simon SLP} (SSLP for short) with respect to the homomorphism $h$ is an SLP such that every path in $\calD(\calG)$
has idempotent-contracted length at most $3 |M|$. If $h$ is clear from the context we simply speak of an SSLP.
A rooted SSLP is called an rSSLP.

\begin{theorem} \label{thm-Simon-SLP}
Fix a finite monoid $M$ and $h : \Sigma^* \to M$ as above.
From a given rSLP $\calG$ one can construct in time $\bigO(|\calG|)$ an rSSLP $\calG'$ of size $\bigO(|\calG|)$ 
such that $\valA{\calG'} = \valA{\calG}$. 
\end{theorem}

\begin{proof}
Recall the nondeterministic transducer $\calT_h$ from Theorem~\ref{thm nondet transducer}.
Let $\calG = (V, \rho, S)$.
In a first step, we compute an acyclic context-free grammar $\calH$ that generates the finite set
$\calT_h(\valA{\calG})$. This is done by the standard product construction applied to the transducer 
$\calT_h$ and the rSLP $\calG$ (viewed as a context-free grammar). 
The variables of $\calH$ are triples $(p,\alpha,q)$, where $\alpha \in V \cup \Sigma$ 
and $p$ and $q$ are states of $\calT_h$. If $\rho(A) = \alpha \beta$ then
for all states $p,q,r$ of $\calT_h$ we have the production $(p,A,r) \to (p,\alpha,q)(q,\beta,r)$ in $\calH$.
Moreover, we add all productions $(p, a, q) \to w$, where $(p, a/w, q)$ is a transition of $\calT_h$. 
If $m$ is the number of states of $\calT_h$ and $t$ is the sum of the lengths of the output words in the transitions of $\calT_h$
 (both are constants in our setting) then 
the size of $\calH$ is bounded by $m^3 |\calG|+t$.

In a second step, we reduce $\calH$ by removing variables that do not derive a terminal word. The well know algorithm from 
\cite[Section~7.1.1]{HoUl79} achieves this in linear time. After this step, we can fix for every variable $A$ an arbitrary
production $A \to s$ and define an SLP by setting $\rho(A) = s$. The resulting rSLP $\calH'$ produces
a word $\gamma(T)$, where $T$ is a factorization tree for $\valA{\calG}$ of height at most $3|M|$. By  Lemma~\ref{lemma-normalization} we can assume
that all right-hand sides of $\calH'$ have length two.

Finally, we restructure the SLP $\calH'$ into an rSSLP $\calG'$ for $\valA{\calG}$.
Every factor $u$ of  the word $\valA{\calH'} = \gamma(T)$ can be uniquely written as
\begin{equation} \label{eq-factor-gamma(T)}
u_1 \, ) \, u_2 \, )  \cdots   u_ k \, ) \, v \, ( \, w_\ell \,  ( \, w_{\ell-1}  \cdots ( \, w_1
\end{equation}
with $k, \ell \leq 3|M|$.
Every word $u_i$, $v$, $w_i$ is either empty, $\gamma(T')$ for 
a subtree $T'$ of $T$ or a sequence $\gamma(T_1) \gamma(T_2) \cdots \gamma(T_j)$ for subtrees $T_1, T_2, \ldots, T_j$ 
whose roots form a consecutive sequence of children of an idempotent node of $T$.
We call \eqref{eq-factor-gamma(T)} the \emph{Dyck factorization} of $u$.

Assume now that a variable $A$ of $\calH'$ derives the word \eqref{eq-factor-gamma(T)}. We then introduce in the final
rSSLP $\calG'$ variables that derive the projections of the words  $u_1, u_2, \ldots, u_ k, v, w_\ell, w_{\ell-1}, \ldots, w_1$ to the alphabet $\Sigma$
 (i.e., the words that are obtained by removing the brackets from $u_1, u_2, \ldots, u_ k, v, w_\ell, w_{\ell-1}, \ldots, w_1$).
 Of courses, it suffices to introduce variables only for those words that are non-empty.
We compute the right-hand side mapping $\rho'$ of $\calG'$ by a bottom-up process. 

Let us consider an example, where $\rho_{\calH'}(A) = BC$ with $B$ and $C$ variables of $\calH'$ (the case
where one of them or both are terminal symbols is easier). Assume that the Dyck factorizations of
the factors $\valGA{\calH'}{B}$ and $\valGA{\calH'}{C}$ of $\gamma(T)$ are
\begin{eqnarray}   
\valGA{\calH'}{B}  & = & r_1 \, ) \, r_2 \, )   \, s \, ( \, t_4 \,  ( \, t_3 \,  ( \, t_2 \,  ( \, t_1, \label{dyck-B} \\
\valGA{\calH'}{C}  & = & x_1 \, ) \, x_2 \, ) \, x_3 \, ) \, y \, ( \, z_2 \,  ( \, z_1. \label{dyck-C}
\end{eqnarray}
Assume that all factors $r_i, s, t_i, x_i, y, z_i$ are nonempty.
Then we have already introduced variables 
$R_1, R_2, S, T_4,  T_3, T_2, T_1$ for the variable $B$
and $X_1, X_2, X_3, Y, Z_2, Z_1$ for $C$. Here the variable $R_i$ derives the $\Sigma$-projection of $r_i$, $S$ derives the  $\Sigma$-projection of  $s$, etc.

The Dyck factorization of $\valGA{\calH'}{A} = \valGA{\calH'}{B}\valGA{\calH'}{C}$ can be derived from the concatenation of the Dyck factorizations in \eqref{dyck-B} and  \eqref{dyck-C}
by pairing matching brackets:
\begin{alignat*}{3} 
 & & r_1 \, ) \, r_2 \, )   \, s \, ( \, t_4 \,  ( \, t_3 \,  ( \, t_2 \,  \textcolor{blue}{( \, t_1}  \!&\,  \,    \textcolor{blue}{x_1 \, )} \, x_2 \, ) \, x_3 \, ) \, y \, ( \, z_2 \,  ( \, z_1 & & \\   
= & & r_1 \, ) \, r_2 \, )   \, s \, ( \, t_4 \,  ( \, t_3 \,  \textcolor{blue}{( \, t_2} \, &   \textcolor{blue}{p_1 \, x_2 \, )} \, x_3 \, ) \, y \, ( \, z_2 \,  ( \, z_1 & & \text{with } p_1 \, = \, ( \, t_1  \, x_1  \, ) \\
= &  & r_1 \, ) \, r_2 \, )   \, s \, ( \, t_4 \,  \textcolor{blue}{( \, t_3} \,  & \textcolor{blue}{p_2} \, \textcolor{blue}{x_3 \, )} \, y \, ( \, z_2 \,  ( \, z_1 & & \text{with } p_2  \, =  \, ( \, t_2  \, p_1  \, x_2  \, ) \\
= &  & r_1 \, ) \, r_2 \, )   \, s \, ( \, \textcolor{blue}{t_4} \,   & \textcolor{blue}{p_3} \, \textcolor{blue}{y} \, ( \, z_2 \,  ( \, z_1 & & \text{with } p_3  \, =  \, ( \, t_3  \, p_2  \, x_3  \,) \\
= &  & r_1 \, ) \, r_2 \, )   \, s \, (  \, & \textcolor{blue}{z_3} \, ( \, z_2 \,  ( \, z_1 & &  \text{with } z_3  \, =  \, t_4 \, p_3 \, y  
\end{alignat*}
We then introduce for $A$ the new variables $P_1, P_2, P_3, Z_3$
with the right-hand sides
\[ \rho'(P_1) = T_1 X_1, \ \rho'(P_2) = T_2 P_1 X_2, \ \rho'(P_3) = T_3 P_2 X_3, \  \rho'(P_4) = T_4 P_3 Y . \]
The variables $R_1, R_2, S, Z_2, Z_1$ can be reused for $A$. Note that the number of new variables introduced
for each variable of $\calH'$ is linearly bounded in the height of the factorization tree $T$; hence it is $\bigO(|M|)$.

At the end we have to apply Lemma~\ref{lemma-normalization} to reduce the length of right-hand sides to two.
It is easy to see that this construction yields an rSSLP. For every variable $A$ of 
$\calG'$ with $\rho'(A) = \alpha_1 \alpha_2$ exactly one of the following two cases holds:
\begin{itemize}
\item There is a binary node $v$ of $T$ with left child $v_1$ and right child $v_2$ such that
$\alpha_1$ produces the word $w(v_1)$ and $\alpha_2$ produces the word $w(v_2)$ (and hence $A$ produces $w(v)$).
\item There is an idempotent node $v$ of $T$ and consecutive children $v_1 <_v v_2 <_v \cdots <_v v_k$ of $v$ (these
are not necessarily all children of $v$) such that $A$ produces $w(v_1) \cdots w(v_k)$, $\alpha_1$ produces $w(v_1) \cdots w(v_i)$ and 
$\alpha_2$ produces $w(v_{i+1}) \cdots w(v_k)$ for some $k \in [1,k-1]$.
\end{itemize}
This ensures that $\calG'$ is indeed an rSSLP.
\end{proof}

\section{Constant time traversal in SLP-compressed words} \label{sec-traversal}

In order to exploit Simon SLPs algorithmically we need a 
 technique from \cite{LohreyMR18} for two-way constant time traversal in SLP-compressed words.
Later, we introduce a generalization of this technique.

Throughout this section we 
fix an SLP $\calG = (V,\rho)$ over the terminal alphabet $\Sigma$.
For a variable $A \in V$ we write $|A|_{\calG}$ 
for the length of the word $\valGA{\calG}{A}$.
Recall from Section~\ref{sec-slp}  the bijection $\pos_A : \Path(A) \to [1,|A|_{\calG}]$ between 
paths in $\calD(\calG)$ that start in $A$ and end in a terminal symbol and the positions in $[1, |A|_{\calG}]$.
We succinctly represent a path $\pi \in \Path(A)$ by merging successive edges where $\pi$ moves in the same direction
(left or right) towards the leaf.
To formalize this idea, we define for every $\alpha\in V \cup \Sigma$ the strings
$L(\alpha), R(\alpha) \in V^*\Sigma$
inductively as follows. For $a \in \Sigma$ let
\begin{gather*}
L(a)  =  R(a) = a.
\end{gather*}
For $A \in V$ with $\rho(A) = \alpha\beta$ ($\alpha, \beta \in V \cup \Sigma$) let
\begin{equation} \label{L-and-R}
L(A)  =  A \, L(\alpha) \text{ and }
R(A)  = A \, R(\beta) .
\end{equation}
 If $L(\alpha)$ (with $\alpha \in V \cup \Sigma$) ends with the terminal symbol $a \in \Sigma$ then we define $\omega_L(\alpha) = a$ (in particular, $\omega_L(a)=a$).
The terminal $\omega_R(\alpha) \in \Sigma$ is defined
analogously by referring to the string $R(\alpha)$.

\begin{example} \label{ex-traversal-string}
Consider the SLP from Example~\ref{ex slp}.
We have $L(a) = R(a)=a$, $L(b)=R(b) = b$, and 
\begin{alignat*}{2}
L(S) & \ = \ SABCa  \qquad \qquad & R(S) & \ = \ SBCDb \\
L(A) & \ = \ ABCa  \qquad \qquad & R(A)& \ = \ ACDb \\
L(B) & \ = \ BCa  \qquad \qquad & R(B) & \ = \ BCDb \\
L(C) & \ = \ Ca \qquad \qquad & R(C) & \ = \ CDb \\
L(D) & \ = \ Da \qquad \qquad & R(D) & \ = \ Db.
\end{alignat*}
Moreover, $\omega_L(X) = a$ and $\omega_R(X) = b$ for all $X \in \{S,A,B,C,D\}$.
\end{example}
We store all strings $L(\alpha)$ (for $\alpha \in V \cup \Sigma$) in $|\Sigma|$ many
tries: Fix $a \in \Sigma$ and let $w_1, \ldots, w_n$ be all strings
$L(\alpha)$ such that $\omega_L(\alpha) = a$ (in particular, $a$ is one of the $w_i$). Let $v_i$ be the string $w_i$ reversed.
Then, $P = \{ v_1, \ldots, v_n \}$ is a prefix-closed set of strings (except that the empty
string is missing) that can be stored
in a trie $T_L(a)$. Formally, $P$ is the set of nodes of $T_L(a)$, each node $s \in P$ is
labeled by its last symbol (so the root is labeled with $a$),
and  there is an edge from $a w$ to $a w A$ for all $w \in V^*$, $A \in V$ with $a w A \in P$.
The tries $T_R(a)$ are defined in the same way by referring to the strings $R(A)$.
Note that the total number of nodes in all tries $T_L(a)$ ($a \in \Sigma$) is exactly $|V|+|\Sigma|$.
In fact, every $\alpha \in V \cup \Sigma$ occurs exactly once as a node label in the
forest $\{ T_L(a) : a \in \Sigma\}$.

\forestset{suffixtree/.style={for tree={edge=->}}}
\begin{figure}[t]
\begin{center}
\begin{forest}
suffixtree,
phantom[
  [$a$ [$C$ [$B$ [$A$[ $S$]]]] [$D$]]
  [$b$]
]
\end{forest}
\qquad
\begin{forest}
suffixtree,
phantom[
  [$a$]
  [$b$ [$D$ [$C$ [$A$] [$B$ [$S$]]]]]
]
\end{forest}
\end{center}
\caption{\label{fig-ex-string} The left shows the tries $T_L(a)$, $T_L(b)$ and the right shows
$T_R(a)$, $T_R(b)$, for the SLP from Example~\ref{ex slp}.}
\end{figure}
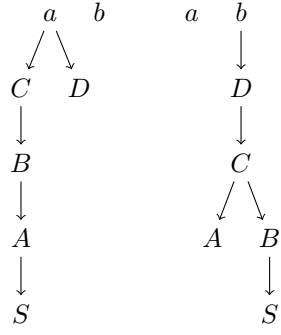

\begin{example}
The tries $T_L(a)$, $T_L(b)$, $T_R(a)$, and $T_R(b)$ for the SLP from Example~\ref{ex slp} are shown
in Figure~\ref{fig-ex-string}.
\end{example}
Next, we define two alphabets $\sfL$ and $\sfR$ by
\begin{align}
\sfL = \ & \{ (A,\ell,\alpha) : \alpha \in \alp(L(A)) \setminus \{A\} \} \label{alph-L}, \\
\sfR = \ & \{ (A,r,\beta) : \beta \in \alp(R(A)) \setminus \{A\} \}. \label{alph-R}
\end{align}
Note that the sizes $|\sfL|$ and $|\sfR|$ are quadratic in the size
of $\calG$. 
In order to avoid some case distinctions in the rest of the section, we also allow triples $(\alpha, d, \alpha)$ 
with $\alpha \in V \cup \Sigma$ and $d \in \{\ell,r\}$ but they are placeholders for the empty word $\varepsilon$ (in particular $(\alpha, \ell, \alpha)$
does not belong to $\sfL$ and similarly for $\sfR$).

On the alphabets $\sfL$ and $\sfR$ we define the functions $\lreduce : \sfL \to \sfL \cup \{\varepsilon\}$ and $\rreduce : \sfR \to \sfR \cup \{\varepsilon\}$ as follows:
let $(A,\ell,\alpha) \in \sfL$ and let $B$ be the unique variable that appears directly to the left of 
$\alpha$ in the string $L(A)$. Then we define $\lreduce(A,\ell,\alpha) = (A,\ell,B)$. Note that by our conventions,
this is $\varepsilon$ in case $A=B$ (i.e., $\alpha$ is the left symbol in $\rho(A)$).
The definition of $\rreduce$ is analogous: If $(A,r,\alpha) \in \sfR$,
then $\rreduce(A,r,\alpha) = (A,r,B)$ where $B$ is the unique variable that appears directly to the left of 
$\alpha$ in the string $R(A)$.

\begin{example} 
For  the SLP from Example~\ref{ex slp} the sets $\sfL$ and $\sfR$ are
\begin{align*}
\sfL = \{ & (S,\ell,A), (S, \ell,B), (S,\ell,C), (S,\ell,a), (A,\ell,B), \\
& (A,\ell,C), (A,\ell,a), (B,\ell,C), (B,\ell,a), (C,\ell,a), (D,\ell,a) \}, \\
\sfR = \{ & (S,r,B), (S,r,C), (S,r,D), (S,r,b), (A,r,C), (A,r,D), \\
& (A,r,b), (B,r,C), (B,r,D), (B,r,b), (C,r,D), (C,r,b),  (D,r,b) \}.
\end{align*}
We have $\lreduce(S,\ell,a) = (S,\ell,C)$, $\rreduce(B,r,D) = (B,r,C)$, and
 $\lreduce(S,\ell,A) = \varepsilon$.
\end{example}
An element $(A,\ell,\alpha)$ can be represented by a pair $(v_1, v_2)$ of different nodes in the forest
$\{ T_L(a) : a \in \Sigma \}$, where $v_1$ (resp. $v_2$) is the unique node labeled with
$\alpha$ (resp., $A$). Note that $v_1$ and $v_2$ belong to the same trie and $v_2$ is strictly below
$v_1$.
This observation allows us to reduce the computation of the mapping $\lreduce$ to
a so-called \emph{next link query}:
From the pair $(v_1, v_2)$ we have to compute the unique child $v$ of $v_1$
such that $v$ is on the path from $v_1$ down to $v_2$. If  $v$ is labeled with $B$, then
$\lreduce(A,\ell,\alpha) = (A,\ell,B)$ (which is $\varepsilon$ if $B=A$). We represent $(A,\ell,B)$ by the pair $(v,v_2)$.
Clearly, the same remark applies to the map $\rreduce$.
The following result is mentioned in~\cite{GasieniecKPS05}.

\begin{proposition}\label{prop:gas}
A trie $T$ can be represented in space $\bigO(|T|)$ such that any next link query can be answered in time
$O(1)$. Moreover, this representation can be computed in time $\bigO(|T|)$ from $T$.
\end{proposition}
We use a compressed representation of paths in the  dag $\calD(\calG)$ that allows to contract consecutive left (resp., right) edges into triples from $\sfL$ (resp., $\sfR$).
Formally, we use sequences of triples
\begin{equation} \label{seq of triples}
\pi = (A_1,d_1, A_2) (A_2,d_2, A_3) \cdots (A_{n-1}, d_{n-1},A_n) (A_n,d_n,\alpha) \in (\sfL \cup \sfR)^+
\end{equation}
such that $n \geq 1$, $A_i \in V$ for $1 \le i \le n$ and $\alpha \in V \cup \Sigma$. 
We say that this $\pi$ starts in $A_1$ and ends in $\alpha$. 
The set of all triple sequences of the form \eqref{seq of triples}
 is denoted with $\Ssf(\calG) \subseteq (\sfL \cup \sfR)^+$.
A sequence $\pi \in \Ssf(\calG)$ represents a path in $\calD(\calG)$ that we denote by $\pat(\pi)$. It is obtained by replacing
every triple $(A_i, \ell, A_{i+1}) \in \sfL$ by a path of consecutive left edges from $A_i$ to $A_{i+1}$ and similarly
for triples from $\sfR$. By definition, every path in $\calD(\calG)$ is also a sequence in $\Ssf(\calG)$ consisting of triples from  
the edge set of $\calD(\calG)$. We will use the term 'path' only in this meaning. If $\pi \in\Ssf(\calG)$ is not a path then we will
always use the term 'sequence'.

If $\pi \in \Ssf(\calG)$ ends in a terminal symbol then $\pi$ is a \emph{terminal sequence}. The set of all terminal sequences
is denoted with $\St(\calG)$.
We say that $\pi \in \Ssf(\calG)$ is \emph{alternating} if for every factor $(A, d, B) (B, d', \alpha)$ in $\pi$ we have 
 $d = \ell$ if and only if $d' =r$. The set of all alternating sequences is denotes with $\Sa(\calG)$.
  An alternating sequence $\pi$ is the maximal compressed sequence of triples for $\pat(\pi)$.
 We write $\Sat(\calG)$ for $\Sa(\calG) \cap \St(\calG)$.
  Finally, for $A \in V$ we write $\Sat(\calG,A), \St(\calG,A), \Sa(\calG,A), \Ssf(\calG,A)$ 
 if we take only sequences starting with the variable $A$. Note that $\Path(\calG,A) \subseteq \St(\calG,A)$.
 For a sequence $\pi \in \St(\calG,A)$ we define the corresponding position in $[1,|A|_{\calG}]$
  by $\pos_A(\pi) = \pos_A(\pat(\pi))$. In the following, the underlying SLP $\calG$ will be clear from the context.
  We will therefore simply write $\Sat(A), \St(A), \Sa(A), \Ssf(A),\Sat, \St, \Sa, \Ssf$.
 
 \begin{example}
 For the SLP from Example~\ref{ex slp}, we have $\pi = (S,\ell,B)(B,r,b) \in \Sat(S)$
 and $\pat(\pi) = (S,\ell,A) (A,\ell,B) (B,r,C) (C,r,D) (D,r,b)$.
\end{example}
There is a simple process of making $\pi \in \Ssf$ alternating. Thereby we replace in $\pi$ factors of the form
$(A,d,B) (B,d,\alpha)$ with $d \in \{\ell,r\}$ by $(A,d,\alpha)$. This can be done in any order and the resulting 
alternating sequence is denoted with $\red(\pi)$. 

In \cite{LohreyMR18} an algorithm is presented that, after a preprocessing phase working in time $\bigO(|\calG|)$,
  takes  as input a sequence 
$\pi \in \Sat(A)$ and computes in constant time the unique $\pi' \in \Sat(A)$
with $\pos_A(\pi') = \pos_A(\pi) + 1$ in case $\pos_A(\pi) < |A|_{\calG}$
and otherwise returns $\bot$ (standing for undefined). A symmetric algorithm that returns the unique $\pi' \in \Sat(A)$
with $\pos_A(\pi') = \pos_A(\pi) - 1$ in case $\pos_A(\pi) > 1$
and otherwise returns $\bot$ is described as well.

Fix a variable $A \in V$ with $|A|_{\calG} \geq 2$ for the further consideration.       
We need a generalization of the algorithms from \cite{LohreyMR18}, where we store a
subset $\Pi \subseteq \Path(A)$ of paths with $|\Pi|$ a \emph{constant}.
It will be convenient to assume that the unique paths $\pi, \pi' \in \Path(A)$ with
$\pos_A(\pi)=1$ and $\pos_A(\pi') = |A|_{\calG}$ belong to $\Pi$; so in particular $|\Pi| \ge 2$.
We store $\Pi$ by a  rooted binary tree $\calB$ (every vertex is either a leaf or has a left and a right child)
with the following properties:
\begin{itemize}
\item $\calB$ has $|\Pi|$ leaves and every edge $e$ of $\calB$ is labelled by some $\pi_e \in \Sa$.
\item If $e \neq e'$ are the two outgoing edges of a node of $\calB$ then 
$\pi_e \in \sfL (\sfL \cup \sfR)^*$ if and only if  $\pi_{e'} \in \sfR (\sfL \cup \sfR)^*$.
\item For every $\pi \in \Pi$ there is a leaf $v$ such that the
following holds: If $e_1, e_2, \ldots, e_m$ are the edges along 
the path from the root of $\calB$ to the leaf $v$, then $\pi_{e_1} \pi_{e_2} \cdots \pi_{e_m} \in \St(A)$ and
$\pat(\pi_{e_1} \pi_{e_2} \cdots \pi_{e_m}) = \pi$.
\end{itemize}
In the situation of the third point, we will also say that the pair $(A,v)$ represents (in the tree $\calB$) the path $\pi$ 
and write $\pos_A(v)$ for $\pos_A(\pi)$. Since $|\Pi| \ge 2$, the tree $\calB$ has at least two leaves.
Since $|\Pi|$ is a constant, $\calB$ has a constant number of nodes. Only the edge labels $\pi_e$
occupy non-constant space.

Note that for given leaves $u,v$ of $\calB$ it is straightforward to decide in time $\mathcal{O}(1)$ which of the cases
$\pos_A(u) < \pos_A(v)$, $\pos_A(u) > \pos_A(v)$, or $\pos_A(u) = \pos_A(v)$ holds. 
The latter holds if and only if $u=v$. If $u \neq v$ then 
one has to follow
the paths from the root of $\calB$ to the leaves $u$ and $v$ up to the point where the two paths diverge to check whether
$\pos_A(u) < \pos_A(v)$ or $\pos_A(u) > \pos_A(v)$ holds.

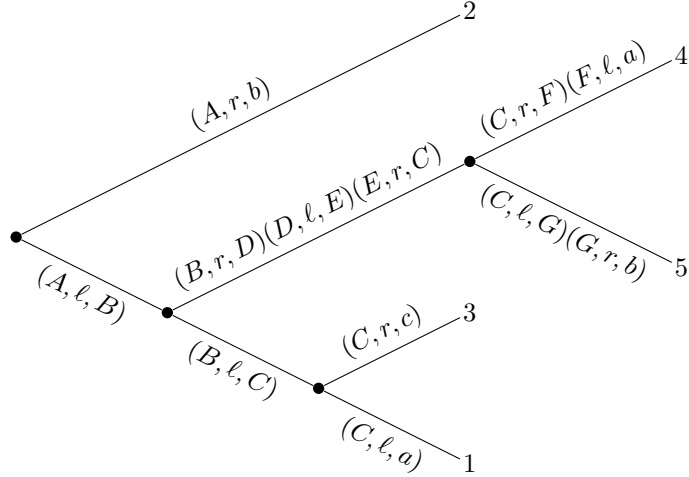
\begin{figure}[t]
 \tikzset{alpha/.style={inner sep = 1pt, fill=white}}
 \tikzset{round/.style={inner sep = 1.5pt, circle, fill=black}}
\centering{
\begin{tikzpicture}

\draw (0,0) node[round] (r) {} 
            ++(6,3)   node[alpha]  (1) {$2$} ; 
\draw (r)  ++(2,-1)  node[round]  (0) {}
               ++(2,-1)  node[round]  (00) {}
               ++(2,-1)  node[alpha]  (000) {$1$} ;
\draw (00) ++(2,1)  node[alpha]  (001) {$3$} ;  
\draw (0) ++(4,2)  node[round]  (01) {}  ++(2.8,1.4)  node[alpha]  (011) {$4$} ;  
\draw (01)  ++(2.8,-1.4)  node[alpha]  (010) {$5$} ;

\draw (r) -- (1) node[midway,sloped,above] {$(A,r,b)$} ;
\draw (r) -- (0) node[midway,sloped,below] {$(A,\ell,B)$} -- (00) node[midway,sloped,below] {$(B,\ell,C)$} -- (000) node[midway,sloped,below] {$(C,\ell,a)$};
\draw (00) -- (001) node[midway,sloped,above] {$(C,r,c)$};
\draw (0) -- (01) node[midway,sloped,above] {$(B,r,D) (D,\ell,E) (E,r,C)$};

\draw (01) -- (010) node[midway,sloped,below] {$(C,\ell,G) (G,r,b)$};
\draw (01) -- (011) node[midway,sloped,above] {$(C,r,F) (F,\ell,a)$};

\end{tikzpicture}}
\caption{An example for a tree $\calB$.}  \label{fig (T,f)}
\end{figure}

Figure~\ref{fig (T,f)} shows an example for a tree
$\calB$ as above. It is not derived from the SLP from Example~\ref{ex slp},
for which the number of left-right and right-left turns in the dag $\calD(\calG)$ is too small to get an interesting $\calB$.
The leaves are called $1,2,3,4,5$ in Figure~\ref{fig (T,f)}.
The reader can easily come up with an SLP that could realize the tree $\calB$ from Figure~\ref{fig (T,f)}.

Consider a tree $\calB$ that stores the set of paths $\Pi \subseteq \Path(A)$.
In the following, $\calB$ is considered as a global
data structure that is modified as a side effect by the procedures that we outline.
The procedure \textsf{right} from Algorithm~\ref{right} in the appendix takes a leaf $u$ of $\calB$ as input. 
If $\pos_A(u) = |A|_{\calG}$ then the algorithm returns $\bot$ and the tree $\calB$ is not modified. 
If  $\pos_A(u) < |A|_{\calG}$ 
and there is a leaf $u'$ in $\calB$ with $\pos_A(u') = \pos_A(u)+1$ then this leaf $u'$ is returned and $\calB$ is not modified.
Otherwise a new leaf $u'$ is added to $\calB$ such that for the new tree we have 
$\pos_A(u') = \pos_A(u)+1$, whereas for all leaves $v$ from the original tree $\calB$, $\pos_A(v)$ does not change.
Moreover, $u'$ is returned.
The procedure \textsf{right} uses a procedure \textsf{split} that can be also found in the appendix.

There is a symmetric procedure \textsf{left} that has the same specification as \textsf{right} except that 
the node $u'$ satisfies $\pos_A(u') = \pos_A(u)-1$ if $\pos_A(u) >  1$; otherwise $\bot$ is returned.
 It can be obtained analogously to the procedure  \textsf{right}.

The idea for the procedure \textsf{right} is derived from  \cite{LohreyMR18}.
 Figure~\ref{fig after update} shows a tree that could arise from the call \textsf{right}$(3)$ for the tree from 
 Figure~\ref{fig (T,f)}. For this, $\lreduce(B,\ell,C) \neq \varepsilon$ (or equivalently, $C$ is not the left symbol in $\rho(A)$) must hold. 
 This causes the $(B,\ell,C)$-labelled edge from Figure~\ref{fig (T,f)} to be split into two edges with labels
 $(B,\ell,B')$ and $(B',\ell,C)$, where $(B,\ell,B') = \lreduce(B,\ell,C)$ and $\rho(B') = CF$. 
 The new node that arises from this edge splitting gets a right child (the new leaf 6 in  Figure~\ref{fig after update}) with label $(B',r,F)(F,\ell,a)$. Here we have
 $a = \omega_L(F)$. 
Note that if  $\lreduce(B,\ell,C) = \varepsilon$ (meaning that the triple $(B, \ell, C)$ represents a single edge in $\calD(\calG)$), then
the call $\textsf{right}(3)$ would split the edge labelled with $(B,r,D) (D,\ell,E) (E,r,C)$ in Figure~\ref{fig (T,f)}.
 Dealing with all the possible cases in the procedure \textsf{right} is not difficult but a bit tedious. Each case needs a constant number of modifications
 in the tree $\calB$, which has constant size by our assumption. In addition a constant number of edge labels (i.e., sequences from $\Sa$) are modified, where every modification
 replaces a constant number of triples at the beginning or the end of the sequence; see also Figures~\ref{fig right 1}--\ref{fig right 5} in the appendix.

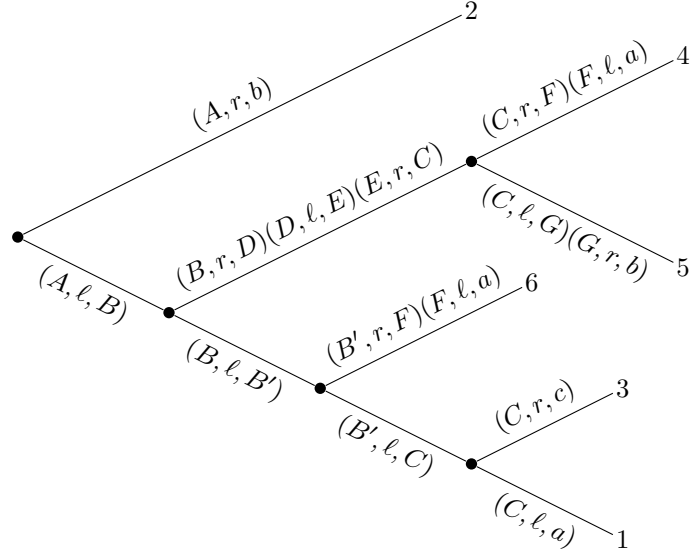
\begin{figure}[t]
 \tikzset{alpha/.style={inner sep = 1pt, fill=white}}
 \tikzset{round/.style={inner sep = 1.5pt, circle, fill=black}}
\centering{
\begin{tikzpicture}

\draw (0,0) node[round] (r) {} 
            ++(6,3)   node[alpha]  (1) {$2$} ; 
\draw (r)  ++(2,-1)  node[round]  (0) {}
               ++(2,-1)  node[round]  (00') {}
               ++(2,-1)  node[round]  (00) {}
               ++(2,-1)  node[alpha]  (000) {$1$} ;
\draw (00) ++(2,1)  node[alpha]  (001) {$3$} ;  
\draw (0) ++(4,2)  node[round]  (01) {}  ++(2.8,1.4)  node[alpha]  (011) {$4$} ;  
\draw (01)  ++(2.8,-1.4)  node[alpha]  (010) {$5$} ;

\draw (00')  ++(2.8,1.4)  node[alpha]  (00'1) {$6$} ;

\draw (r) -- (1) node[midway,sloped,above] {$(A,r,b)$} ;
\draw (r) -- (0) node[midway,sloped,below] {$(A,\ell,B)$} -- (00') node[midway,sloped,below] {$(B,\ell,B')$} 
       -- (00) node[midway,sloped,below] {$(B',\ell,C)$} -- (000) node[midway,sloped,below] {$(C,\ell,a)$};
\draw (00) -- (001) node[midway,sloped,above] {$(C,r,c)$};
\draw (0) -- (01) node[midway,sloped,above] {$(B,r,D) (D,\ell,E) (E,r,C)$};

\draw (01) -- (010) node[midway,sloped,below] {$(C,\ell,G) (G,r,b)$};
\draw (01) -- (011) node[midway,sloped,above] {$(C,r,F) (F,\ell,a)$};

\draw (00') -- (00'1) node[midway,sloped,above] {$(B',r,F) (F,\ell,a)$};

\end{tikzpicture}}
\caption{The tree $\calB$ from Figure~\ref{fig (T,f)} after the update $\textsf{right}(3)$.
We have $\rho(B') = CF$.}  \label{fig after update}
\end{figure}

We will also need a procedure for removing a leaf $v$ from the tree $\calB$ (where $1 < \pos_A(v)< |A|_{\calG}$) in constant time.
This is straightforward.
We first remove the leaf $v$ together with its incoming edge. This leaves two edges $e = (x,y), e' = (y,z)$ where $(y,z)$ is the unique
outgoing edge of $y$. We merge these edges into a single edge with label $\red(\pi_e \pi_{e'}) \in \Sa$. If for instance 
$\pi_e$ ends with a triple $(U,\ell,V)$ and $\pi_{e'}$ starts with a triple $(V,\ell,W)$ then these triples are combined
into the triple $(U,\ell,W)$. 
Figure~\ref{fig remove leaf} shows the tree after removing the leaf $5$ from the tree in Figure~\ref{fig (T,f)}.

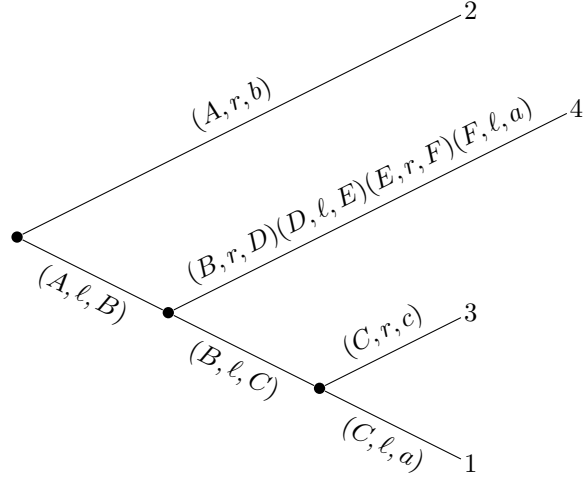
\begin{figure}[t]
 \tikzset{alpha/.style={inner sep = 1pt, fill=white}}
 \tikzset{round/.style={inner sep = 1.5pt, circle, fill=black}}
\centering{
\begin{tikzpicture}

\draw (0,0) node[round] (r) {} 
            ++(6,3)   node[alpha]  (1) {$2$} ; 
\draw (r)  ++(2,-1)  node[round]  (0) {}
               ++(2,-1)  node[round]  (00) {}
               ++(2,-1)  node[alpha]  (000) {$1$} ;
\draw (00) ++(2,1)  node[alpha]  (001) {$3$} ;  
\draw (0) ++(5.4,2.7)  node[alpha]  (011) {$4$} ;  

\draw (r) -- (1) node[midway,sloped,above] {$(A,r,b)$} ;
\draw (r) -- (0) node[midway,sloped,below] {$(A,\ell,B)$} -- (00) node[midway,sloped,below] {$(B,\ell,C)$} -- (000) node[midway,sloped,below] {$(C,\ell,a)$};
\draw (00) -- (001) node[midway,sloped,above] {$(C,r,c)$};
\draw (0) -- (011) node[midway,sloped,above] {$(B,r,D) (D,\ell,E) (E,r,F) (F,\ell,a)$};

\end{tikzpicture}}
\caption{The tree $\calB$ from Figure~\ref{fig (T,f)} after removing the leaf $5$.}  \label{fig remove leaf}
\end{figure}          

 Note that the above tree $\calB$ only stores paths from $\Path(A)$. For our application in Section~\ref{sec-main} we need trees
 $\calB_A$ for several variables $A \in V$. We then write $\textsf{right}(A,u)$ and $\textsf{left}(A,u)$ to indicate the tree $\calB_A$
 on which these operations are executed.

\section{Ranked MSO-enumeration} \label{sec-main}

In this section we prove the main technical result, that covers ranked MSO-enumeration (as explained in the introduction) as well
as the enumeration of the image $f(w)$ for a polyregular function, both for compressed input words.

Fix two finite alphabets $\Sigma$ and $\Gamma$ and an MSO-interpretation 
 $f : \Sigma^+ \to \Gamma^+$ as defined in Section~\ref{sec-polyreg}. 
 Let $\Phi(x_1,\ldots,x_k)$, $\Phi_<(x_1,\ldots,x_k, y_1,\ldots,y_k)$, and $\Phi_b(x_1,\ldots,x_k)$ for all $b \in \Gamma$
 be the corresponding MSO-formulas over words from $\Sigma^+$ defining $f$. We fix these formulas for this section and
define the following additional MSO-formulas (we write $\bar{x}$ for $(x_1, \ldots,x_k)$ and similarly for $\bar{y}$ and $\bar{z}$):
 \begin{eqnarray*}
 \Phi_{\textsf{min}}(\bar{x}) & = & \Phi(\bar{x}) \wedge \neg \exists \bar{z}: \Phi_<(\bar z, \bar{x}), \\
 \Phi_{\textsf{max}}(\bar{x}) & = & \Phi(\bar{x}) \wedge \neg \exists \bar{z}: \Phi_<(\bar{x},\bar z), \\
 \Phi_{\succ}(\bar{x},\bar{y}) & = & \Phi_<(\bar{x},\bar{y}) \wedge \neg \exists \bar{z}: \Phi_<(\bar{x},\bar{z}) \wedge \Phi_<(\bar{z},\bar{y}), \\
\Phi_{\succ,i}(\bar{x},y_i) & = & \exists y_1,\ldots, y_{i-1}, y_{i+1}, \ldots, y_k : \Phi_{\succ}(\bar{x},\bar{y}).
 \end{eqnarray*}
So, for every word $w \in \Sigma^+$ of length $n$, there are unique tuples $\min_w, \max_w \in [1,n]^k$ such that 
$w \models \Phi_{\min}(\min_w)$ and $w \models \Phi_{\max}(\max_w)$ hold.
Clearly, $\min_w$ (resp., $\max_w$) is the first (resp., last) position in
the word structure for $f(w)$. Moreover,  $\Phi_{\succ}(\bar{x},\bar{y})$ defines the successor relation on positions of $f(w)$. 
Hence, the binary relation $\valA{\Phi_{\succ}}_w$ defines a function $\succ_w : \valA{\Phi}_w \setminus \{ \textsf{max}_w \} \to  \valA{\Phi}_w \setminus \{ \textsf{min}_w \}$. 
Similarly, $\Phi_{\succ,i}(\bar{x},y_i)$ defines for every word $w$ of length $n$ a function $\succ_{w,i} : \valA{\Phi}_w \setminus \{ \textsf{max}_w \} \to [1,n]$
that maps a tuple $\bar{p}$ from its domain to the $i$-th component of $\succ_w(\bar{p})$.
Finally, the formulas $\Phi_b$ ($b \in \Gamma$) define a function $\gamma : \valA{\Phi}_w \to \Gamma$, where $\gamma(\bar{p})$ is the unique 
symbol $b \in \Gamma$ such that $w \models \Phi_b(\bar{p})$ holds for $\bar{p} \in \valA{\Phi}_w$.

We then define a function $F$ that maps an input word $w \in \Sigma^+$
to the sequence $(b_1, \bar{p}_1) (b_2, \bar{p}_2) \cdots (b_m, \bar{p}_m)$ with 
$b_i \in \Gamma$ and $\bar{p}_i \in [1,n]^k$ that is defined as follows:
\begin{itemize}
\item
$\bar{p}_1 = \min_w$, $\bar{p}_m = \max_w$,
\item for every $i \in [2,m]$, $\bar{p}_i = \succ_w(\bar{p}_{i-1})$,
\item for every $i \in [1,m]$, $b_i = \gamma(\bar{p}_i)$.
\end{itemize}
So $F$ is the same as the MSO-interpretation $f$, where in addition we add the domain tuples from $f(w)$ to the output.
The main result of this section is the following:

\begin{theorem} \label{thm main}
The enumeration problem that maps an rSLP $\calG$ to the sequence $F(\valA{\calG})$ can be solved 
after linear preprocessing in
constant delay.
\end{theorem}

\begin{proof}
Let $w = \valA{\calG}$.
In order to prove the theorem, we will compute in linear time from the rSLP $\calG$ a data structure that allows to solve in constant time the following tasks:
\begin{enumerate}[(i)]
\item compute the tuple  $\min_w$,
\item compute from a given tuple $\bar{p} \in \valA{\Phi}_w$ the symbol $\gamma(\bar p)$,
\item compute from a given tuple $\bar{p} \in \valA{\Phi}_w \setminus \{ \textsf{max}_w \}$ and $i \in [1,k]$ the position $\succ_{w,i}(\bar p)$, 
\item check whether a given tuple $\bar{p} \in \valA{\Phi}_w$ is $\max_w$.
\end{enumerate}
Point (iii) then allows to compute in constant  time also the tuple $\succ_{w}(\bar p)$ (since $k$ is a constant in our setting).
To compute $F(w)$ one then iterates the computation of $\succ_{w}$  
starting with the $k$-tuple $\min_w$. Moreover, it suffices to show (iii) for $i=1$ (the same method works for all $i \in [1,k]$).

In the following, we focus on step (iii) (the most difficult one) and thereby we will also address step (iv).
Later we explain how steps (i) and (ii) can be solved using similar techniques.
We  will use a factorization tree $T$ for $w$ with respect to a morphism $h : \Sigma^* \to M$, where $h$ and $M$ only
depend on the fixed MSO-interpretation. We start with the definition of $M$ and $h$.

Let $\mathcal{Z} = \{ x_1, \ldots, x_k, y_1\}$ be the set of free variables of the formulas $\Phi_{\succ,1}$.
We transform $\Phi_{\succ,1}$ into a  trim NFA $\calA$ using the standard procedure for transforming MSO-formulas over words into NFAs.
The alphabet of $\calA$ is the  set $\Sigma \times 2^\mathcal{Z}$ and $\calA$ has the following properties:
\begin{itemize}
\item For every word $(a_1,Z_1) (a_2, Z_2) \cdots (a_n,Z_n) \in L(\calA)$ (with $a_i \in\Sigma$ and $Z_i \subseteq \mathcal{Z}$)
 and every variable $z \in \mathcal{Z}$ there is exactly one position $i \in [1,|w|]$ with $z \in Z_i$.
\item  $(a_1,Z_1) (a_2, Z_2) \cdots (a_n,Z_n) \in L(\calA)$ if and only if 
$\succ_{u,1}(p_1,\ldots,p_k) = p'$,
where $u = a_1 a_2 \cdots a_n$, $p_i$ is the unique position in $[1,n]$ with $x_i \in Z_{p_i}$ and $p'$ is the unique position in $[1,n]$ with $y_1 \in Z_{p'}$.
\end{itemize}
This, together with the fact that $\calA$ is trim, ensures that we can partition the state set $Q$ of $\calA$ into two sets $Q_{0}$ and $Q_{1}$ such that the following holds,
where a path in $\calA$ is called $y_1$-free if it  does not use a transition with a label $(a,Z)$ 
such that $y_1 \in Z$:
\begin{itemize}
\item states in $Q_{0}$ can be reached from the initial state $q_{0}$ of $\calA$ along a $y_1$-free path, and
\item from states in $Q_{1}$ one can reach a final state along a $y_1$-free path.
\end{itemize}
Note that this means that on every path from $q_{0}$ to a state in $Q_{1}$ the variable $y_1$ must be seen exactly once
and, similarly, on every path from $Q_{0}$ to a final state the variable $y_1$ must be seen exactly once.
For a transition $(q, (a,Z), q')$ in $\calA$ with $y_1 \in Z$ we must have $q \in Q_0$ and $q' \in Q_1$.

From $\calA$ we then derive the NFA $\calA'$ by replacing every transition label $(a,Z)$ 
with $y_1 \in Z$  by $(a, Z \setminus \{y_1\})$.
It is an NFA for the MSO-formula $\exists y_1 : \Phi_{\succ,1} \ = \ \exists \bar y : \Phi_{\succ}(\bar x, \bar y)$. 
Its alphabet is $\Sigma \times 2^{\{x_1,\ldots,x_k\}}$.
Note that $\valA{\exists \bar y : \Phi_{\succ}(\bar x, \bar y)}_w = \valA{\Phi}_w \setminus \{ \max_w \}$.
We define the monoid $M = \Bool(\calA', \Sigma \times \{\emptyset\})$. Let
$h : \Sigma^* \to M$ be the homomorphism with $h(a) = B_{a}$, where $B_a$ is the boolean matrix corresponding
to the state transformation induced by the letter $(a,\emptyset)$ in the NFA $\calA'$. Note that transitions in $\calA'$ having a label
$(a, Z)$ with $Z \neq\emptyset$ are irrelevant for the monoid $M$ (one could remove all these transitions and
would obtain the same monoid).

\subparagraph{The uncompressed setting.}
Before we proceed with the SLP-compressed setting, let us first consider the case where $w$ is an uncompressed string of length $n$.
This case is sketched in \cite{Boj18}. We follow this approach with some adaptions and later show how to extend it to the compressed setting
using the machinery from Sections~\ref{sec SSLP} and \ref{sec-traversal}.

Let  $w = a_1 a_2 \cdots a_n$ with $a_i \in \Sigma$. In the precomputing phase 
we compute in time $\bigO(n)$ a factorization tree $T$ of height at most $3|M|$ for the word $w$ with respect to the above homomorphism $h : \Sigma^* \to M$.
We assume that $T$ is stored by having pointers from every node to its parent node, it left sibling (if it exists) and its right sibling (if it exists).
Moreover, there are pointers from every non-leaf node to its left-most child and to its right-most child.
Let $v_0$ be the root of $T$. For every node $v$ of $T$ let $|v|_T$ be the number of leaves of $T$ below $v$
(the so-called \emph{leaf size} of $v$). We precompute in time $\bigO(|w|)$ all numbers $|v|_T$ and all monoid elements $h(v)$
for $v$ a node of $T$. For a subset $U$ of nodes of $T$ we write $h(U)$ for the product $\prod_{u \in U} h(u)$, where
we run over the elements from $U$ in the order $<_T$.

Let $\bar{p} = (p_1, p_2, \ldots, p_k) \in [1,n]^k$ be the argument for $\succ_{w,1}$ and assume that $p_1 < p_2 < \cdots < p_k$. The case
where the order between the $p_i$ is different or some of the $p_i$ are equal does not lead to additional complications. 
In addition, it will be convenient to assume that $p_1 = 1$ and $p_k = n$. This can be assumed by adding dummy variables in the MSO-formulas
that are fixed to the first and last position of the word $w$.
We assume that $\bar{p} \in \valA{\Phi}_w$; this property will be always preserved during the computation of $F(w)$.
We denote with $w \otimes \bar{p}$ the word $(a_1, Z_1) (a_2, Z_2) \cdots (a_n,Z_n)$, where $Z_i = \{ x_j : j \in [1,k], p_j = i \}$. 
Our goal is to check whether $\bar{p} \neq \max_w$ (see step (iv)) and, if the latter holds,
compute $p' = \succ_{w,1}(\bar{p}) \in [1,n]$. 

Let $v_i$ be the $p_i$-th leaf of $T$ with respect to the order $<_T$. 
We make the additional assumption that for every $i \in [1,k]$ we have computed the path $\pi_i = [v_0,v_i]$ from the root of $T$ to the leaf $v_i$;
this property will be also preserved during the computation of $F(w)$.
Moreover, for every node $v$ on one of these paths $\pi_i$ we assume that we have computed the 
number
\begin{equation} \label{v_T<}
|v|_T^< := | \{ u : u \text{ is a leaf of $T$ and } u <_T v \}|.
\end{equation}
For the root node $v_0$ we have $|v_0|_T^< = 0$.
We will compute in addition to $p' = \succ_{w,1}(\bar{p})$  (if it exists) also the path $\pi' = [v_0, v']$, where $v'$ is the $p'$-th leaf of $T$ with respect to the order $<_T$. The computation starts in $v_0$
and ends in the leaf $v'$.
Moreover, for every node $v$ along this path $\pi'$ we compute the number $|v|_T^<$. When this process finally arrives at the leaf $v'$, we can
determine the position $p'$ by $p' = |v'|_T^<+1$.

The constant height bound $3 |M|$ of the tree $T$ ensures that one can compute all monoid elements $h(w[p_{i-1}+1, p_i-1]) \in M$ ($2 \le i \leq k$) in constant time. To see this,
let $v_{i-1,i}$ be the lowest common ancestor of the nodes $v_{i-1}$ and $v_i$ and define the paths $\tau_{i-1} = [v_{i-1,i},v_{i-1}]$ and $\tau_i = [v_{i-1,i}, v_i]$; see also
 Figure~\ref{fig part fact tree}. 
We define the following set of nodes of $T$:
 \begin{eqnarray*}
  \encl_{i-1,i} &=& \{ v : v_{i-1} <_T v <_T v_i \} \setminus \tau_i
\end{eqnarray*}
($\encl$ stands for `enclosed').
In Figure~\ref{fig part fact tree}, $\encl_{i-1,i}$ is the gray shaded part without
the red paths $\tau_{i-1}$ and $\tau_i$.
 For every node $u$ in $T$ we define $\ch_{i-1,i}(u)$ as the set of those children of $u$ that belong to $\encl_{i-1,i}$.
 We have $\ch_{i-1,i}(u) = \emptyset$ if $u$ does not belong to $\tau_{i-1}  \cup \encl_{i-1,i} 
 \cup \tau_i$. Moreover, define 
 \[ \ch_{i-1,i} = \bigcup_{u \in  \tau_{i-1} \cup \tau_i} \ch_{i-1,i}(u) .
 \]
 In Figure~\ref{fig part fact tree}, $\ch_{i-1,i}$ is the set of (dark and light) blue nodes in the gray region.
  Then we can compute 
 $h(w[p_{i-1}+1, p_i-1])$ as 
 \begin{equation} \label{M-product over branching-off nodes}
 h(w[p_{i-1}+1, p_i-1]) \ = h(\ch_{i-1,i}).
 \end{equation}
  This product can be computed in constant time using the 
 already computed paths $\pi_{i-1}$ and $\pi_i$, the pointers in $T$, and the monoid elements $h(v)$. 
 Powers of the same idempotent element $e$ in $h(\ch_{i-1,i})$ can be reduced to a single $e$. Then the 
 resulting product has at most one monoid element for each node in $\tau_{i-1} \cup \tau_i$, which yields
 a product of length at most $6 |M|$.
 For the example, in Figure~\ref{fig part fact tree} we obtain 
 $h(\ch_{i-1,i}) = m\,n\,e\,o\,m\,\epsilon \, m\,\epsilon\,o\,n\,e\,o$.

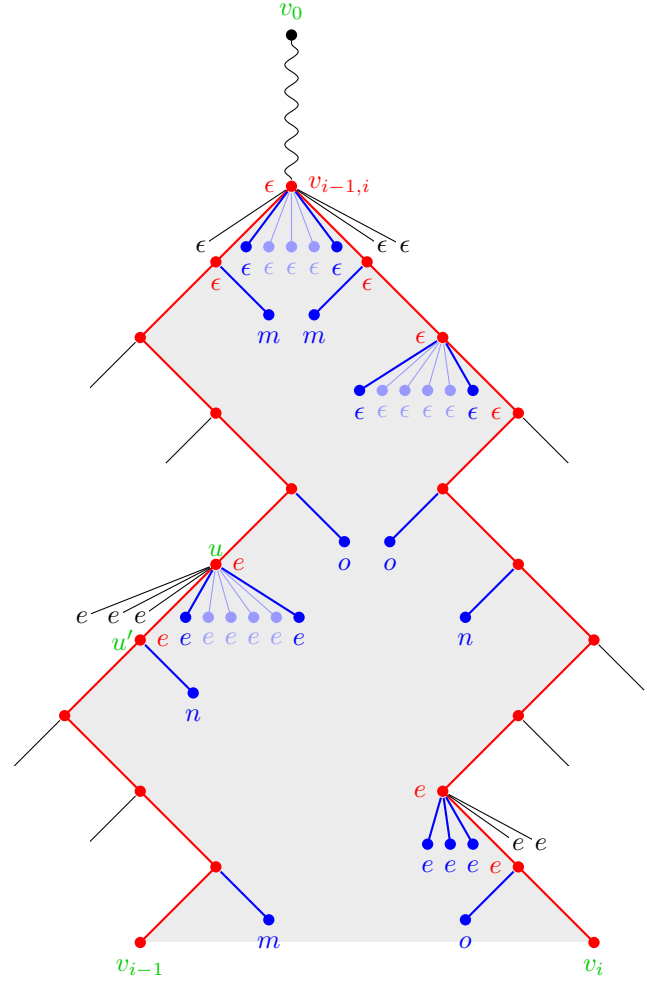
\begin{figure}[t]
 \tikzset{alpha/.style={inner sep = 1pt, fill=white}}
 \tikzset{round/.style={inner sep = 1.5pt, circle, fill=black}}
\centering{
\begin{tikzpicture}

\draw[white, fill=gray!15] (0,0) -- ++(1,1)  -- ++(-1,1) -- ++(-1,1)  -- ++(1,1)  -- ++(1,1) -- ++(1,1) -- ++(-1,1) -- ++(-1,1) -- ++(1,1)  -- ++(1,1) 
-- ++(1,-1) -- ++(1,-1) -- ++(1,-1) -- ++(-1,-1) -- ++(1,-1) -- ++(1,-1) -- ++(-1,-1) -- ++(-1,-1) -- ++(1,-1) -- ++(1,-1) -- cycle;

\draw[red,thick] (0,0) node[round, color=red, label = below:\textcolor{darkgreen}{$v_{i-1}$}] (v1) {} 
           -- ++(1,1)  node[round,red]  (x1) {}
           -- ++(-1,1)  node[round,red]  (x2) {} 
           -- ++(-1,1)  node[round,red]  (x3) {} 
           -- ++(1,1)  node[round,red, label = right:\textcolor{red}{$e$}, label= {[label distance=-1.4mm]left:$\textcolor{darkgreen}{u'}$}]  (x4) {} 
           -- ++(1,1)  node[round,red, label = right:\textcolor{red}{$e$}, label = {[label distance=-1.1mm]above:$\textcolor{darkgreen}{u}$}]  (x5) {} 
           -- ++(1,1)  node[round,red]  (x6) {} 
           -- ++(-1,1)  node[round,red]  (x7) {} 
           -- ++(-1,1)  node[round,red]  (x8) {} 
           -- ++(1,1)  node[round,red, label = below:\textcolor{red}{$\epsilon$}]  (x9) {} 
           -- ++(1,1)  node[round, color=red, label = right:\textcolor{red}{$v_{i-1,i}$},label = left:\textcolor{red}{$\epsilon$}]  (v) {} 
           -- ++(1,-1)  node[round,red, label = below:\textcolor{red}{$\epsilon$}]  (z1) {} 
           -- ++(1,-1)  node[round,red, label = left:\textcolor{red}{$\epsilon$}]  (z2) {} 
            -- ++(1,-1)  node[round,red, label = left:\textcolor{red}{$\epsilon$}]  (z3) {}
             -- ++(-1,-1)  node[round,red]  (z4) {} 
             -- ++(1,-1)  node[round,red]  (z5) {} 
             -- ++(1,-1)  node[round,red]  (z6) {} 
             -- ++(-1,-1)  node[round,red]  (z7) {}
             -- ++(-1,-1)  node[round,red, label = left:\textcolor{red}{$e$}]  (z8) {}  
             -- ++(1,-1)  node[round,red, label = left:\textcolor{red}{$e$}]  (z9) {} 
             -- ++(1,-1)  node[round, color=red, label = below:\textcolor{darkgreen}{$v_i$}]  (v2) {}    ;

\draw[decorate,decoration=snake]  (v) -- ++(0,2)  node[round, label = above:$\textcolor{darkgreen}{v_0}$] {}; 
             
\draw[blue,thick] (x1) -- ++(.7,-.7)  node[round, color=blue, label = below:\textcolor{blue}{$m$}]  {} ;
\draw[blue,thick] (x4) -- ++(.7,-.7)  node[round, color=blue, label = below:\textcolor{blue}{$n$}]  {} ;
\draw (x5) -- ++(-1.77,-.7)  node[alpha]  {$e$} ;
\draw (x5) -- ++(-1.35,-.7)  node[alpha]  {$e$} ;
\draw (x5) -- ++(-1,-.7)  node[alpha]  {$e$} ;
\draw[blue,thick] (x5) -- ++(-.4,-.7)  node[round, color=blue, label = below:\textcolor{blue}{$e$}]  {} ;
\draw[lightblue] (x5) -- ++(-.1,-.7)  node[round, color=lightblue, label = below:\textcolor{lightblue}{$e$}]  {} ;
\draw[lightblue] (x5) -- ++(.2,-.7)  node[round, color=lightblue, label = below:\textcolor{lightblue}{$e$}]   {} ;
\draw[lightblue] (x5) -- ++(.5,-.7)  node[round, color=lightblue, label = below:\textcolor{lightblue}{$e$}]   {} ;
\draw[lightblue] (x5) -- ++(.8,-.7)  node[round, color=lightblue, label = below:\textcolor{lightblue}{$e$}]   {} ;
\draw[blue,thick] (x5) -- ++(1.1,-.7)  node[round, color=blue, label = below:\textcolor{blue}{$e$}]  {} ;
\draw[blue,thick] (x6) -- ++(.7,-.7)  node[round, color=blue, label = below:\textcolor{blue}{$o$}] {} ;
\draw[blue,thick] (x9) -- ++(.7,-.7)  node[round, color=blue, label = below:\textcolor{blue}{$m$}]   {} ;

\draw (x2) -- ++(-.7,-.7)  node[alpha]   {} ;
\draw (x3) -- ++(-.7,-.7)  node[alpha]   {} ;
\draw (x7) -- ++(-.7,-.7)  node[alpha]   {} ;
\draw (x8) -- ++(-.7,-.7)  node[alpha]   {} ;

\draw(v) -- ++(-1.2,-.8)  node[alpha]  {$\epsilon$} ;
\draw[blue,thick] (v) -- ++(-.6,-.8)  node[round, color=blue, label = below:\textcolor{blue}{$\epsilon$}]  {} ;
\draw[lightblue] (v) -- ++(-.3,-.8)  node[round, color=lightblue, label = below:\textcolor{lightblue}{$\epsilon$}]  {} ;
\draw[lightblue] (v) -- ++(0,-.8)  node[round, color=lightblue, label = below:\textcolor{lightblue}{$\epsilon$}]  {} ;
\draw[lightblue] (v) -- ++(.3,-.8)  node[round, color=lightblue, label = below:\textcolor{lightblue}{$\epsilon$}] {} ;
\draw[blue,thick] (v) -- ++(.6,-.8)  node[round, color=blue, label = below:\textcolor{blue}{$\epsilon$}] {} ;
\draw(v) -- ++(1.2,-.8)  node[alpha]  {$\epsilon$} ;
\draw(v) -- ++(1.5,-.8)  node[alpha]  {$\epsilon$} ;

\draw[blue,thick] (z1) -- ++(-.7,-.7)  node[round, color=blue, label = below:\textcolor{blue}{$m$}] {} ;

\draw[blue,thick] (z2) -- ++(-1.1,-.7)  node[round, color=blue, label = below:\textcolor{blue}{$\epsilon$}] {} ;
\draw[lightblue] (z2) -- ++(-.8,-.7)  node[round, color=lightblue, label = below:\textcolor{lightblue}{$\epsilon$}]  {} ;
\draw[lightblue] (z2) -- ++(-.5,-.7)  node[round, color=lightblue, label = below:\textcolor{lightblue}{$\epsilon$}]  {} ;
\draw[lightblue] (z2) -- ++(-.2,-.7)  node[round, color=lightblue, label = below:\textcolor{lightblue}{$\epsilon$}] {} ;
\draw[lightblue] (z2) -- ++(.1,-.7)  node[round, color=lightblue, label = below:\textcolor{lightblue}{$\epsilon$}]  {} ;
\draw[blue,thick] (z2) -- ++(.4,-.7)  node[round, color=blue, label = below:\textcolor{blue}{$\epsilon$}]  {} ;
\draw[blue,thick] (z4) -- ++(-.7,-.7)  node[round, color=blue, label = below:\textcolor{blue}{$o$}] {} ;
\draw[blue,thick] (z5) -- ++(-.7,-.7)  node[round, color=blue, label = below:\textcolor{blue}{$n$}] {} ;

\draw[blue,thick] (z8) -- ++(-.2,-.7)  node[round, color=blue, label = below:\textcolor{blue}{$e$}]   {};
\draw[blue,thick] (z8) -- ++(.1,-.7)  node[round, color=blue, label = below:\textcolor{blue}{$e$}]   {};
\draw[blue,thick] (z8) -- ++(.4,-.7)  node[round, color=blue, label = below:\textcolor{blue}{$e$}]   {};
\draw (z8) -- ++(1,-.7)  node[alpha]   {$e$};
\draw (z8) -- ++(1.3,-.7)  node[alpha]   {$e$};

\draw[blue,thick] (z9) -- ++(-.7,-.7)  node[round, color=blue, label = below:\textcolor{blue}{$o$}]  {} ;

\draw (z3) -- ++(.7,-.7)  node[alpha]  {} ;
\draw (z6) -- ++(.7,-.7)  node[alpha]  {} ;
\draw (z7) -- ++(.7,-.7)  node[alpha]  {} ;

\end{tikzpicture}}
\caption{A part of a factorization tree $T$. The paths $\tau_{i-1} = [v_{i-1,i},v_{i-1}]$ and $\tau_i = [v_{i-1,i}, v_i]$ are in red. 
Node names are in green.
Some of the nodes $u$ are labelled with their monoid elements $h(u) \in M$ ($m,n,o$ and the idempotents $e,\epsilon$).
The blue nodes in the gray zone are in general not leaves of $T$; the subtrees rooted in those nodes are omitted.
When the path $\pi'$ branches off from $\pi_{i-1}$ or $\pi$ into the gray zone, it can only take one of the dark blue edges.}
\label{fig part fact tree}
\end{figure}

Next, observe that the monoid elements $h(w[p_{i-1}+1, p_i-1]) \in M$ ($2 \le i \leq k$) together with the symbols $w[p_i]$ 
($1 \le i \leq k$) are enough to check whether $\succ_{w,1}(\bar{p})$ is defined.
Using these elements one can check,  whether there is an accepting run of the NFA $\calA'$
on the word $w \otimes \bar{p}$.
If there is such an accepting run, then $p' = \succ_{w,1}(\bar{p}) \in [1,n]$ is defined, otherwise
we must have $\bar{p} = \max_w$ and the computation of $F(w)$ terminates. 

Assume that $p' = \succ_{w,1}(\bar{p})$ is defined. We have to compute this position $p'$ and the corresponding path $\pi'$.
First, we find out to which of the intervals $\{p_1\}, [p_1+1, p_2-1], \{p_2\}, \ldots, [p_{k-1}+1, p_k-1], \{p_k\}$ 
the unknown position $p'$ belongs to (recall that we assume $p_1=1$ and $p_k = n$).
For this, the partition $Q = Q_{0} \uplus Q_{1}$ of the state set $Q$ of the NFA
$\calA'$ is important. If we find for instance an accepting run of 
the NFA $\calA'$ on $w \otimes \bar{p}$ that enters the interval  $[p_{i-1}+1, p_i-1]$ in a state from $Q_{0}$ and 
leaves $[p_{i-1}+1, p_i-1]$ in a state from $Q_{1}$ (the existence of such a run  can be deduced from $h(w[p_{i-1}+1, p_i-1])$)
then $p' \in [p_{i-1}+1, p_i-1]$. 

If it turns out that $p' = p_i$ for some $i \in [1,k]$ then we are done with step (iii). Therefore, 
assume that  $p_{i-1} < p' < p_i$.
The path $\pi' = [v_0,v']$ shares with $\pi_{i-1}$ and $\pi_i$ the prefix $[v_0,v_{i-1,i}]$.
Then $\pi'$ follows either $\tau_{i-1}$ or $\tau_i$, but at some point it will 
enter $\encl_{i-1,i}$
 (the path can already branch off at $v_{i-1,i}$ if $v_{i-1,i}$ has more than two children as in Figure~\ref{fig part fact tree}). 
 The node $u \in \tau_{i-1} \cup \tau_i$ from where $\pi'$ enters 
 $\encl_{i-1,i}$ can be determined from the monoid elements 
 $h(\ch_{i-1,i}(v))$ for $v \in \tau_{i-1} \cup \tau_i$ using again the partition $Q = Q_{0} \uplus Q_{1}$. 
When this node $u$ is determined we have to find in constant time the right child from $\ch_{i-1,i}(u)$ that is taken by the path $\pi'$.
If $u$ is an idempotent node, then the set $\ch_{i-1,i}(u)$ can have unbounded size.
The important observation here is the following: Assume that $u \in \tau_{i-1} \cup \tau_i$ is an idempotent node.
Then only the edge to the first or last (with respect to the child order $<_u$) node of $\ch_{i-1,i}(u)$
can be taken by $\pi'$ (unless $|\ch_{i-1,i}(u)|=3$, in which case all three nodes are possible),
and the correct alternative  can be found in time $\bigO(1)$ using the $h$-values of the nodes of $T$ 
and the partition $Q = Q_{0} \uplus Q_{1}$.
Assume for instance the situation shown in Figure~\ref{fig idempotent nodes} on the left occurs, where the unique position of
$y_1$ is located at a leaf below the third node from $\ch_{i-1,i}(u)$.
Hence, one can obtain an accepting run of the NFA 
$\calA$ (the NFA for $\Phi_{\succ,1}$) on the word $w$ for the situation on the left.
But then the same must also hold for the situation on the right in Figure~\ref{fig idempotent nodes},
because the second and third node from $\ch_{i-1,i}(u)$ 
have the same context in the monoid $M$ ($(ee, eee) = (e,e)$ on the left and $(e, eeee) = (e,e)$ on the right in Figure~\ref{fig idempotent nodes}) is the same.
Hence, we could also find for $y_1$ a position below the second (or fourth) 
node from $\ch_{i-1,i}(u)$.
This contradicts the uniqueness of the position of $y_1$. 
Note that this argument assumes that $\ch_{i-1,i}(u)$ consists of at least four nodes.
In Figure~\ref{fig part fact tree} the path $\pi'$ can only branch off into the gray zone to one of the dark blue nodes.
Once we found the right child of $u$ we can navigate down to the leaf $v'$ (where $y_1$ is located) using the same principle.

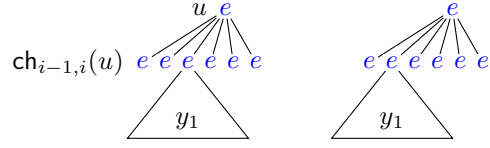
\begin{figure}[t]
 \tikzset{alpha/.style={inner sep = 1pt, fill=white}}
 \tikzset{round/.style={inner sep = 1.5pt, circle, fill=black}}
  \tikzset{empty/.style={inner sep = 0pt}}
\centering{
\begin{tikzpicture}
\draw (3,0) node[alpha, label = left:$u$] (z) {\textcolor{blue}{$e$}} ;
\draw (z) -- ++(-1.1,-.7)  node[alpha, label = left:$\ch_{i-1,i}(u)$]  (m1) {\textcolor{blue}{$e$}} ;
\draw (z) -- ++(-.8,-.7)  node[alpha]  (m2) {\textcolor{blue}{$e$}} ;
\draw (z) -- ++(-.5,-.7)  -- ++(.8,-1) -- ++(-1.6,0) node[midway,above] {$y_1$} -- ++(.8,1) node[alpha]  (m3) {\textcolor{blue}{$e$}} ;
\draw (z) -- ++(-.2,-.7)  node[alpha]  (m4) {\textcolor{blue}{$e$}} ;
\draw (z) -- ++(.1,-.7)  node[alpha]  (m5) {\textcolor{blue}{$e$}} ;
\draw (z) -- ++(.4,-.7)  node[alpha]  (m6) {\textcolor{blue}{$e$}} ;

\draw (6,0) node[alpha] (z) {\textcolor{blue}{$e$}} ;
\draw (z) -- ++(-1.1,-.7)  node[alpha]  (m1) {\textcolor{blue}{$e$}} ;
\draw (z) -- ++(-.8,-.7)  -- ++(.8,-1) -- ++(-1.6,0) node[midway,above] {$y_1$} -- ++(.8,1) node[alpha]  (m2) {\textcolor{blue}{$e$}} ;
\draw (z) -- ++(-.5,-.7)  node[alpha]  (m3) {\textcolor{blue}{$e$}} ;
\draw (z) -- ++(-.2,-.7)  node[alpha]  (m4) {\textcolor{blue}{$e$}} ;
\draw (z) -- ++(.1,-.7)  node[alpha]  (m5) {\textcolor{blue}{$e$}} ;
\draw (z) -- ++(.4,-.7)  node[alpha]  (m6) {\textcolor{blue}{$e$}} ;
\end{tikzpicture}}
\caption{The variable $y_1$ cannot be located at a leaf of $T$ as shown on the left, because then $y_1$ could be also
located at a leaf of $T$ as shown on the right, contradicting the uniqueness of the position of $y_1$. Similarly, the situation on the right is excluded, otherwise also
the situation on the left would be possible.}  \label{fig idempotent nodes}
\end{figure}

Updating the numbers $|u|_T^<$ while we navigate down to $v'$ is straightforward (using the precomputed leaf sizes $|v|_T$) 
when we move from a binary node $u$ to one of its children. 
When descending from an idempotent node $u$ to one of its children in $\ch_{i-1,i}(u) = \{u_1, u_2, \ldots, u_{l-1}, u_l\}$ (with $u_i <_u u_{i+1}$)
 we have seen that only $u_1, u_2, u_l$ are possible for the descent
(this also includes the case where $l = 3$). Moreover, $u_1$ (resp., $u_l$) is either the left-most (resp., right-most) child of $u$ in $T$
or the right (resp., left) sibling of another child $u'$ of $u$ that belongs to $\tau_{i-1}$ (resp., $\tau_i$).
In all cases, the $| \cdot |_T^<$-value of the new node can be easily computed from previously computed numbers. For instance, if the algorithms descends to 
$u_l$ and $u_l$ is the left sibling of a node $u' \in \tau_i$ then we can determine $|u_l|_T^<$ as $|u'|_T^< - |u_l|_T$ (note that $|u'|_T^<$ has been computed 
before since it belongs to $\pi_i$)
and if the algorithms descends to 
$u_2$ and $u_1$ is the right sibling of a node $u' \in \tau_{i-1}$ then we can determine $|u_2|_T^<$ as $|u'|_T^< + |u'|_T + |u_1|_T$.
This concludes the outline of steps (iii) and (iv) (from the beginning of the proof of Theorem~\ref{thm main}) for the uncompressed case.

The other steps are in fact easier. For step (ii) we transform each of the formulas $\Phi_b$ ($b \in \Gamma$) from our fixed MSO-interpretation
 into an NFA $\calA_b$ and take the monoid $M_b = \Bool(\calA_b, \Sigma \times \{\emptyset\})$ and the homomorphism
 $h_b : \Sigma^* \to M_b$ defined by
 $h_b(a) = B_{b,a}$, where $B_{b,a}$ is the boolean matrix corresponding
to the state transformation induced by the letter $(a,\emptyset)$ in the NFA $\calA_b$. Then, in order to check 
whether $w \models \Phi_b(\bar p)$ holds, it suffices to compute all the monoid elements 
$h_b(w[p_{i-1}+1, p_i-1]) \in M_b$ ($2 \le i \leq k$) using a factorization tree for $w$ with respect to $h_b$.
This task has been solved for step (iii) with $h$ instead of $h_b$.
Finally, step (i) can be solved in the same way as step (iii):
We compute each component of the tuple $\min_w$ using the formula $\Phi_{\min,i}(x_i) = \exists x_1, \ldots, x_{i-1}, x_{i+1}, \ldots, x_k : \Phi_{\min}$.
 For this one can use the algorithm from step (iii). The only difference is 
 that there are no positions $p_i$ and paths $\pi_i$ to start with. Formally, one can simply take $k=0$ in the algorithm for step (iii). 

Instead of working with a separate factorization tree for each of the MSO-formulas $\Phi_{\min,i}$, $\Phi_{\succ,i}$ ($1 \le i \le k$), $\Phi_b$
($b \in \Gamma$) it is better to take the direct product 
of the above monoids $M, M_b$, etc., and the corresponding homomorphism from $\Sigma^*$ into this direct product.
Then one can work with a single factorization tree $T$. This has the advantage that the path to the new position $p'$ computed 
in step (iii) does not have to be inserted into the factorization trees for the other MSO-formulas.  This will be crucial when
$w$ is given by an rSLP.

\subparagraph{The dag-compressed setting.}
Our final goal is to adapt the above approach to the setting where the string $w$ is 
given by an rSLP $\calG$. As an intermediate step, let us consider the case, where the factorization tree
$T$ for $w$ is represented succinctly by a rooted dag, where isomorphic subtrees of $T$ are identified. 
To represent the child orders $<_u$, one can label the outgoing edges of a dag node by $1,2,\ldots, k$
where $k$ is the number of children (for the binary dag $\calD(\calG)$ of 
an SLP these edge labels are $\ell$ and $r$).
We write $\calD$ for this dag and its root is denoted with $r$.

Consider for instance the factorization tree $T$ from Figure~\ref{fact tree}. Then a dag compression of $T$ is shown in 
Figure~\ref{dag fact tree} (it is in fact the best dag compression in the sense that all nodes of $T$, where isomorphic subtrees
are rooted are merged in the dag).
We store $\calD$ by keeping for every dag node $u$ a doubly linked adjacency list of the children of $u$
with additional pointers from $u$ to the first and last node of this list.

\begin{figure}[t]
 \tikzset{alpha/.style={inner sep = 1pt, fill=white}}
 \tikzset{round/.style={inner sep = 1.5pt, circle, fill=black}}
  \tikzset{empty/.style={inner sep = 0pt}}
\centering{
\begin{tikzpicture}
\node[round,magenta] (0) {} ;
\node[round, magenta, below left = 1 and 2 of 0] (1) {};
\node[round, darkgreen, below right = 1 and 2 of 0] (2) {};
\node[round, darkgreen, below = 1 of 1] (3) {};
\node[alpha, below right = 1 and 1 of 1] (d) {\textcolor{magenta}{$d$}};

\node[alpha, below = 2 of 3] (4) {\textcolor{magenta}{$c$}};
\node[round, magenta, below = 3 of 3] (5) {};
\node[alpha, below left = 1 and .5 of 5] (a) {\textcolor{magenta}{$a$}};
\node[alpha, below right = 1 and .5 of 5] (b) {\textcolor{magenta}{$b$}};
\node[alpha, below right = 1 and .6 of 2] (c) {\textcolor{magenta}{$e$}};

\draw[->] (0) to node[above = -.6mm]{$\scriptstyle{1}$} (1);
\draw[->] (0) to node[above = -.6mm]{$\scriptstyle{2}$} (2);

\draw[->] (1) to node[left = -.6mm]{$\scriptstyle{1}$} (3);
\draw[->] (1) to node[above = -.6mm]{$\scriptstyle{2}$} (d);

\draw[->] (3) to[bend right =20] node[left = -1mm]{$\scriptstyle{2}$} (4);
\draw[->] (3) to[bend left=20]  node[right = -1mm]{$\scriptstyle{3}$} (4);
\draw[->] (3) to[bend right =40] node[left = -1mm]{$\scriptstyle{1}$} (5);
\draw[->] (3) to[bend left = 40]  node[right = -1mm]{$\scriptstyle{4}$} (5);

\draw[->] (5) to node[left = -.6mm]{$\scriptstyle{1}$} (a);
\draw[->] (5) to node[right = -.6mm]{$\scriptstyle{2}$} (b);

\draw[->] (2) to node[above = -.6mm]{$\scriptstyle{1}$} (5);
\draw[->] (2) to[bend left = 30] node[below = -.6mm]{$\scriptstyle{2}$} (4);
\draw[->] (2) to[bend left = 50] node[below = -.6mm]{$\scriptstyle{3}$} (5);
\draw[->] (2) to node[right = -.6mm]{$\scriptstyle{4}$} (c);

\end{tikzpicture}}
\caption{A dag compression of the factorization tree from Figure~\ref{fact tree}. Binary nodes are in magenta, idempotent nodes are in green.}  \label{dag fact tree}
\end{figure}
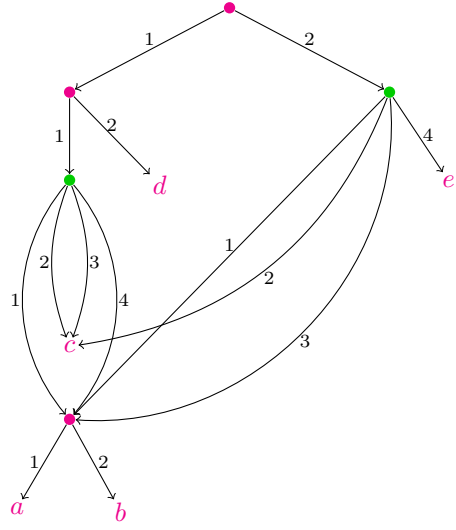

We claim that the approach for the uncompressed setting extends to the dag-compressed setting 
with only a few minor modifications.
Nodes of $T$ are in one-to-one correspondence with paths in $\calD$ that start in the root and end in a leaf (we speak of root-leaf paths in the following).
Such a path can be identified with the sequence of edge labels which yields a lexicographic order on paths. 
The algorithm follows the approach from the uncompressed setting to compute the path $\pi'$ to 
the unique position of the variable $y_1$ in the word $w$ from the paths $\pi_1, \ldots, \pi_k$ (that lead to the positions
of the variables $x_1, \ldots, x_k$). In the dag-compressed setting, these paths are of course paths in the dag
$\calD$. 

Recall that in the uncompressed setting the computation of the monoid elements from $M$ and the decision on the branching direction
when descending along the path $\pi'$ are solely based on the precomputed monoid elements $h(u)$ for the nodes $u$ of $T$.
Now, the nodes of $T$ are paths $\pi$ that start at the root of $\calD$; so we should write $h(\pi)$ instead of $h(u)$. In general
the number of paths in $\calD$ is exponential in the size of $\calD$; so we cannot afford to compute all values $h(\pi)$.
But this is not necessary, since $h(\pi)$ only depends on the last node of $\pi$.
More precisely, for every node $v$ of $\calD$ we can precompute in time $\bigO(|\calD|)$ 
a monoid element $h(v)$ as follows: If $v$ is a leaf labelled with the symbol $a \in \Sigma$ then $h(v) = h(a)$ and if $v$ has
$m \geq 1$ outgoing edges 
$(v, i, v_i)$ ($1 \leq i \leq m$) then $h(v) = h(v_1) h(v_2) \cdots h(v_m)$ in $M$. 
Then for every path $\pi$  in $\calD$ that goes from the root to
 $v$ we have $h(\pi) = h(v)$ (where $\pi$ is viewed here as a node of the tree $T$).
The same remark also applies to the leaf size $|\pi|_T$ for a node $\pi$ of $T$. Recall that
this is the number of leaves in the subtree rooted at $\pi$. Also these values only depend on the last
node of $\pi$ and we can precompute numbers $|v|_{\calD}$ (for $v$ a node of $\calD$) by
$|v|_{\calD}=1$ for a leaf $v$ of $\calD$ and 
$|v|_{\calD} = |v_1|_{\calD} + |v_2|_{\calD} + \cdots + |v_m|_{\calD}$ if $v$ has $m \geq 1$ outgoing edges $(v,i,v_i)$.
Using these monoid elements $h(v)$ and the numbers $|v|_{\calD}$ the computation of the path $\pi'$
 from the uncompressed setting
can be extended without modifications to the dag-compressed setting. Thereby we also compute
for every prefix $\tau$ of $\pi'$  the value  $|\tau|_T^<$; see \eqref{v_T<}. Note that $|\tau|_T^<$ is the number of root-leaf paths
in the dag $\calD$ that are lexicographically strictly smaller than $\tau$.

\subparagraph{The SLP-compressed setting.}
We finally adapt the above approach for the dag-compressed setting to the SLP-compressed setting, where $w$ is given by an rSLP $\calG$. 
Let $w = \valA{\calG}$ and $N = |w|$ in the following. We only consider step (iii) from the beginning of the proof in detail.
In the preprocessing  we replace, using Theorem~\ref{thm-Simon-SLP}, the rSLP $\calG$ by an equivalent rSSLP with respect to the above homomorphism $h : \Sigma^* \to M$. We denote this rSSLP again with $\calG$. So, in the following, $\calG = (V,\rho,S)$ is an rSSLP with $\valA{\calG}=w$.
For an idempotent $e \in E(M)$ let $\calG_e = (V_e, \rho\rest_{V_e})$ be the $e$-part of $\calG$ (see Section~\ref{sec SSLP}).
In time $\bigO(|\calG|)$ we precompute for every variable $A \in V$ the monoid element 
$h(A) := h(\valGA{\calG}{A}) \in M$ as well as the length $|A|_{\calG}$ of the string produced by $A$.

The rSSLP $\calG$ is quite close to a dag $\calD$ for a factorization tree $T$ (as considered in the paragraph on the dag-compressed setting).
The only difference is that in the rSSLP $\calG$ an idempotent node $v$ of $\calD$ (i.e., a node with more than two children) 
becomes a variable of the $e$-part $\calG_e$  for the idempotent $e = h(v)$. Hence, a node of $\calD$ with more than two children
 is replaced by a binary subdag. 
 Using the data structure from Section~\ref{sec-traversal} we can emulate the algorithm for the dag-compressed setting on the rSSLP $\calG$.
 Let us elaborate this in more detail.

As in the dag-compressed setting,
the paths $\pi_i$ in the factorization tree $T$ from the uncompressed setting are replaced by terminal paths $\pi_i \in \Path(S)$  
in the dag $\calD(\calG)$. In contrast to the dag-compressed setting (where the dag has depth at most $3 |M|$ -- the same as the factorization tree to which it unfolds)
the paths $\pi_i$ in the rSLP $\calG$ have unbounded length.
In order to be able to do all manipulations in constant time,
we use the data structure from Section~\ref{sec-traversal} for maximal subpaths of $\pi_i$ that belong to an $e$-part $\calG_e$.

Consider a tuple $\bar\pi = (\pi_1, \ldots, \pi_k)$ of paths  $\pi_i \in \Path(S)$ in the dag $\calD(\calG)$.
 Let $\bar{p} = (p_1, \ldots, p_k)$ be the corresponding tuple of positions $p_i = \pos_S(\pi_i)$ in the word $w$.
 We assume that $\bar{p}$ is also computed and that it belongs to the domain of $\succ_w$.
 Our goal is to compute the position $p' = \succ_{w,1}(\bar{p})$ and the corresponding path $\pi' \in \Path(S)$.

 For every idempotent $e \in E(M)$ let $\Path(\bar\pi,e)$ be the set of all terminal paths $\sigma \in\Path(\calG_e)$ in $\calG_e$ 
 such that for some $i \in [1,k]$ there is a factorization $\pi_i = \tau_1 \sigma \tau_2$ (where $\tau_1$ or $\tau_2$ can be also empty), but
 if if $\tau_1$ is non-empty and ends with the edge $(A,d,B)$ then $(A,d,B) \tau \notin\Path(\calG_e)$.\footnote{The same is also true if we add
 to $\sigma$ the first edge of $\tau_2$ since $\sigma$ is required to be a terminal path in $\calD(\calG_e)$.} 
 In other words, $\sigma$ is a longest subpath in some $\pi_i$ under the restriction that it is a path in $\calG_e$.
For every variable $A$ of $\calG_e$ let $\Path(\bar\pi,A)$ be those paths in $\Path(\bar\pi,e)$ that start in the variable $A$
 (we do not have to mention $e$ in $\Path(\bar\pi,A)$, since $A$ determines the idempotent $e = h(A)$).
 
 For every variable $A$ of $\calG_e$ such that $\Path(\bar\pi,A) \neq \emptyset$
 we will store the paths in $\Path(\bar\pi,A)$ by a tree $\calB_A$ as described in 
Section~\ref{sec-traversal}. Note that the union of all sets $\Path(\bar\pi,A)$ is bounded by $3 |M| k$:
Each of the $k$ paths $\pi_i$ contains at most $3|M|$ maximal subpaths from some $e$-part $\calG_e$, since the idempotent
contracted length of $\pi_i$ is bounded by $3|M|$.
Hence, we will store at most $3 |M| k$ many trees $\calB_A$. 

With the help of the $\calB_A$ we store every path $\pi_i$ by replacing in $\pi_i$ every maximal subpath 
$\sigma \in \Path(\bar\pi,A)$ (such subpaths cannot overlap in $\pi_i$) by an edge labelled with the pair $(A,v)$
that represents the path  $\sigma$ in the sense of Section~\ref{sec-traversal} (in particular, $v$ is a leaf of $\calB_A$).
We denote the resulting contracted path by $\tilde\pi_i$ and call it the \emph{contracted representation} of $\pi_i$.
 Its length (i.e., number of edge triples $(A,d,B)$ and pairs $(A,v)$)
is at most $3 |M|$, since $\calG$ is an rSSLP
and hence has idempotent-contracted depth at most $3 |M|$. 
The paths stored in the trees $\calB_A$ correspond in the dag-compressed setting to edges in a path $\pi_i$ that lead from an idempotent node to 
one of its children. In Figure~\ref{fig:id-contract}, the shortcuts from $A_3$ to $A_5$ and from $A_7$ to $A_9$ 
would be placed by pairs $(A_3,v)$ and $(A_7,v')$ for suitable leaves $v$ and $v'$ in $\calB_{A_3}$ and $\calB_{A_7}$, respectively.

\begin{figure}[t]
 \tikzset{alpha/.style={inner sep = 1pt, fill=white}}
 \tikzset{round/.style={inner sep = 1.5pt, circle, fill=black}}
\centering{
\begin{tikzpicture}
\draw[red] (0,0) node[alpha] (1) {$a$} 
           -- ++(-1,1) node[pos=.6,right] {$r$} node[alpha]  (2) {$F$} ++(0,2) node[alpha]  (3) {$E$} ;         
           
\draw[decorate,decoration=snake,red]  (3) --   node[midway,left] {$u$} (2) ;
\draw[decorate,decoration=snake,blue] (3) --++ (.5,-2) node[alpha]  {$F_1$};
\draw[decorate,decoration=snake,lightblue] (3) --++ (1,-2) node[alpha]  {$F_2$};
\draw[decorate,decoration=snake,lightblue] (3) --++ (1.5,-2) node[alpha]  {$F_3$};
\draw[decorate,decoration=snake,blue] (3) --++ (2,-2) node[alpha]  {$F_4$};

\draw[red] (3) -- ++(1,1) node[midway,left] {$\ell$} node[alpha]  (4) {$D$}
  -- ++(1,1) node[pos=.6,left] {$\ell$} node[alpha]  (5) {$B$};
  
\draw[red] (5) ++(1.4,1.4) node[alpha]  (6) {$A$} ;
\draw (6)-- ++(1,1) node[pos=0.5,left] {$\ell$} node[alpha]  (7) {$S$} ;
\draw[decorate,decoration=snake,red]  (6)  -- node[pos=.3,left=1mm] {$v$} (5);
\draw[decorate,decoration=snake,red]  (6) -- ++(1.4,-1.4)  node[pos=.3, right= 2mm] {$v'$} node[alpha]  (8) {$C$} ;  

\draw[decorate,decoration=snake,blue] (6) --++ (-.7,-1.4) node[alpha]  {$A_1$};
\draw[decorate,decoration=snake,blue] (6) --++ (0,-1.4) node[alpha]  {$A_2$};
\draw[decorate,decoration=snake,blue] (6) --++ (.7,-1.4) node[alpha]  {$A_3$};

\draw[red] (8) -- ++(1,-1) node[pos=0.3,right] {$r$} node[alpha]  (9) {$E$} ;
\draw[decorate,decoration=snake,red]  (9) -- ++(0,-3) node[midway,right] {$u'$}  node[alpha]  (10) {$G$} ;
\draw[decorate,decoration=snake,blue]  (9) -- ++(-.5,-3)   node[alpha]  (11) {$G_1$} ;
\draw[decorate,decoration=snake,lightblue]  (9) -- ++(-1,-3)   node[alpha]  (11) {$G_2$} ;
\draw[decorate,decoration=snake,lightblue]  (9) -- ++(-1.5,-3)   node[alpha]  (11) {$G_3$} ;
\draw[decorate,decoration=snake,blue]  (9) -- ++(-2,-3)   node[alpha]  (11) {$G_4$} ;
\draw[red] (10) -- ++(1,-1) node[pos=0.3,right] {$r$} node[alpha]  (11) {$b$} ;
\draw[blue] (4)  -- ++(1,-1) node[pos=0.3,right] {$r$} node[alpha]  {$D'$} ; 
\draw[blue] (5)  -- ++(1,-1) node[pos=0.3,right] {$r$} node[alpha]  {$B'$} ; 
\draw[blue] (8)  -- ++(-1,-1) node[pos=0.3,left] {$\ell$} node[alpha]  {$C'$} ; 
\draw[blue] (10)  -- ++(-1,-1) node[pos=0.3,left] {$\ell$} node[alpha]  {$G'$} ; 
\end{tikzpicture}}
\caption{The contracted path $\tilde{\pi}_{i-1}$ goes from $S$ down to the $a$-labelled leaf, and the 
contracted path $\tilde{\pi}_{i}$ goes from $S$ down to the $b$-labelled leaf on the bottom right.}  \label{fig compute h(w)}
\end{figure}

\begin{example}
Figure~\ref{fig compute h(w)} shows an example for the contracted paths 
\begin{eqnarray}
\tilde{\pi}_{i-1} &=& (S,\ell,A) (A,v) (B,\ell, D) (D,\ell,E) (E,u) (F,r,a) \text{ and } \label{pair (E,i)}\\
\tilde{\pi}_{i} &=& (S,\ell,A) (A,v') (C,r,E) (E,u') (G,r,b)
\end{eqnarray}
with $v \neq v'$. The straight edges in Figure~\ref{fig compute h(w)} are edges in the dag $\calD(\calG)$
and the curly edges stand for paths from $\bigcup_{e \in E(M)} \Path(\calG_e)$.
The red $u$-labelled path from $E$ to $F$ is represented by the path from the 
 root to the leaf $u$ in the tree $\calB_E$, and similarly for the other red curly edges. 
 The curly paths from $E$ to $F_i$ ($1 \leq i \leq 4$) are the paths from $\Path(\calG_e,E)$ (assuming that $E \in V_e$)
 that are lexicographically larger than the path from $E$ to $F$ and the lexicographic order on these paths corresponds
 to the left-to-right order in  Figure~\ref{fig compute h(w)}.\footnote{Note that the variables $F, F_1, \ldots, F_4$ are not necessarily distinct.}
Analogous statements hold for the curly paths that descend from $A$ and the $E$ on the right. 
  \end{example}
Using the trees $\calB_A$, we can now emulate
the algorithm for the dag-compressed setting. Note that in the dag-compressed setting we assumed that the children of an idempotent dag node $u$ are 
stored in a doubly linked adjacency list. 
This allows to go from a child $u'$ of $u$ to its left or right sibling, which is needed when we branch off
from one the paths $\pi_i$; see also Figure~\ref{fig part fact tree} (which refers to the uncompressed setting, but the situation for the dag-compressed setting
is the same), where the path $\pi'$ can branch off from $\pi_{i-1}$ or $\pi_i$ along one of the dark blue edges.  
For instance, in Figure~\ref{fig part fact tree} the path $\pi'$ might branch off from $\pi_{i-1}$ at the idempotent node $u$ and 
 descend to the right sibling of $u'$. Also notice that the nodes $u$ and $u'$
have been visited before, since they belongs to $\pi_{i-1}$ or $\pi_i$. 
For the SLP-compressed setting, this means that during the computation of the new path $\pi'$, the latter may
branch off from some $\pi_i$ along a \emph{terminal path} $\tau' \in \Path(\calG_e,A)$ 
such that $\tau'$ is the lexicographic predecessor or successor
 (with respect to the lexicographic order on $\Path(\calG_e,A)$)
of a path that is already represented in $\calB_A$.
In Figure~\ref{fig compute h(w)} the path $\pi'$ may for instance branch off at $A$ and continue the dark blue path to $A_1$.
Its lexicographical predecessor is the red curly path from $A \in V_e$ to $B$ in $\calD(\calG_e)$,
which is already represented by the pair $(A,v)$ in the tree $\calB_A$. Hence, the call
$\mathsf{right}(A,v)$ returns a potentially new leaf $v'$ in the new tree $\calB_A$
such that the pair $(A,v_1)$ represents
in the new tree $\calB_A$ (after the call $\mathsf{right}(A,v)$) the blue path from $A$ to $A_1$ in Figure~\ref{fig compute h(w)}.
At this point we have computed the prefix $(S,\ell,A) (A,v_1)$ of $\tilde\pi'$ (the contracted representation of the path $\pi'$).

In general, the path $\pi'$ can only branch off from $\pi_{i-1} \cup \pi_i$ along one the dark blue edges in 
Figure~\ref{fig compute h(w)} (assuming that $p_{i-1} < p' < p_i$). The case where $\pi'$ branches off along the curly edge from $E$ to $F_4$ for instance
is easy, since this edge stands for the rightmost path in $\Path(\calG_{e'},E)$ (assuming $E \in V_{e'}$). This path is already represented in the tree $\calB_E$
(we assumed in Section~\ref{sec-traversal} that the leftmost and rightmost path in the SLP are always represented in the tree $\calB$).

We have now shown that after a $\bigO(|\calG|)$-time preprocessing we can compute the mapping $\succ_{w}$ (for $w = \valA{\calG}$) in constant time.
 More precisely, given a tuple $\bar{p} = (p_1, \ldots, p_k)$ of positions in $w$ and the contracted representations of the paths $\pi_i$ (the lexicographically $p_i$-th path
 in $\Path(\calG,S)$), our algorithm computes in constant time the new tuple $\bar{p}'$ of positions
 in $w$. Moreover, the algorithm also computes the contracted representations of the paths $\pi'_i$ that correspond to the new positions $p'_i$ 
 and updates the trees $\calB_A$ accordingly. It is now important to remove those leaves in the trees $\calB_A$ that are not
 needed for the contracted representations of the paths $\pi'_i$, otherwise the trees $\calB_A$ would grow unboundedly. 
 For this we have to do a constant number (at most $3 |M| k$ many) 
 delete operations as explained at the end of Section~\ref{sec-traversal}.
 
 Finally, recall the last paragraph from the uncompressed setting, where we argued that one should take the direct product of all the monoids for the various
 MSO-formulas $\Phi_{\min,i}$, $\Phi_{\succ,i}$ ($1 \le i \le k$), $\Phi_b$
($b \in \Gamma$)
 and compute a single factorization tree for this direct product. 
 For the SLP-compressed setting this means that we work with a single rSSLP. This
is crucial. If one would work with several rSSLPs (one for each of the above MSO-formulas), then after computing the new path $\pi'$ to the new position
$p'$ in the rSSLP for $\Phi_{\succ,1}$, one would have to compute in the rSSLP for $\Phi_{\succ,2}$ the lexicographically $p'$-th terminal path.
It is not clear how to do this in constant time, even if one works with contracted paths.
 
This finally concludes the proof of Theorem~\ref{thm main}. 
\end{proof}
By removing the position tuples $\bar{p}$ from the function $F$ in Theorem~\ref{thm main}, we obtain the following corollary:
\begin{corollary} \label{coro polyreg}
Fix a polyregular function $f : \Sigma^* \to \Gamma^*$.
The enumeration problem that maps an rSLP $\calG$ over the terminal alphabet $\Sigma$ to the word 
$f(\valA{\calG})$ can be solved after linear preprocessing in constant delay.
\end{corollary}
By removing the letters $b \in \Gamma$ (or taking a unary alphabet $\Gamma$) we obtain the following corollary of Theorem~\ref{thm main}:
 \begin{corollary}
 Fix MSO-formulas $\Phi(x_1, \ldots, x_k)$ and $\Phi_{<}(x_1, \ldots, x_k,y_1,\ldots,y_k)$ such that for every
 word $w \in \Sigma^*$ the relation $\valA{\Phi_<}_w$ is a linear order on $[1,|w|]^k$. 
 Then, the enumeration problem that maps an rSLP $\calG$ to the sequence of all tuples from $\valA{\Phi}_{\valA{\calG}}$ ordered
 with respect to $\valA{\Phi_{<}}_{\valA{\calG}}$ can be solved after linear preprocessing in constant delay.
 \end{corollary}

 \begin{proof}
 We apply Theorem~\ref{thm main} with a unary alphabet $\Gamma$ and the MSO-interpretation that is given by the formulas 
 $\Phi(\bar{x})$ for the domain and $\Phi'_{<}(\bar{x},\bar{y}) = \Phi_{<}(\bar{x},\bar{y}) \wedge \Phi(\bar{x}) \wedge \Phi(\bar{y})$
 for the linear order.
 \end{proof}
 The reader might have observed that our proof of Theorem~\ref{thm main}
also works, if we directly start with a formula  $\Phi_{\succ}(\bar{x},\bar{y})$
 for a successor relation on $k$-tuples instead of
 a formula  $\Phi_{<}(\bar{x},\bar{y})$ for a linear order on $k$-tuples. This leads to the notion of 
 \emph{successor-MSO string-to-string interpretations} \cite{BojanczykKL19}, which
 are the string-to-string functions computed by MSO-interpretations, assuming that a string $a_1 a_2 \cdots a_n$ is represented by its
 successor structure $([1,n], p \mapsto p+1 (1 \le p < n), (P_a)_{a \in \Sigma})$.
 Hence, we obtain:
 \begin{corollary} \label{coro polyreg-ext}
Fix a successor-MSO string-to-string interpretation  $f : \Sigma^* \to \Gamma^*$.
The enumeration problem that maps an rSLP $\calG$ over the terminal alphabet $\Sigma$ to the word 
$f(\valA{\calG})$ can be solved after linear preprocessing in constant delay.
\end{corollary}
The class of successor-MSO string-to-string interpretations is strictly larger than the class of polyregular functions. Moreover, in contrast
to the polyregular functions, the class of successor-MSO string-to-string interpretations is not closed under composition and membership 
in inverse images of regular languages is undecidable \cite{BojanczykKL19}.

\section{Future work}

We proved that for an MSO-definable order on $k$-tuples, the ranked enumeration problem for a fixed MSO-query and a given grammar-compressed string
$w$ can be solved after linear time preprocessing (where `linear' means `linear in the size of the SLP for $w$') in constant delay.
Here, the MSO-query contains $k$ free first-order variables but no free set variables. This is a restriction that we would like to overcome.
Many existing results on MSO-query enumeration allow queries with free set variables  \cite{AmarilliBMN19,Bagan06,BourhisGJR21,GMS24,LohreyS26,MunozRiveros2025}, in which case `constant delay' has to be replaced
by `output-linear delay'. At the moment, we do not see, how to extend our approach based on factorization trees to free set variables. 

Another desirable extension of our work would be to go from strings to forests. 
Only very recently a version of Simon's factorization tree theorem for forests has been shown \cite{ACS26}
using the framework of forest algebras. It remains to see whether this result can be combined with forest
straight-line programs (in the same way as factorization trees for strings were made accessible to ordinary
straight-line programs in Section~\ref{sec SSLP}). Constant-time traversal algorithms for forest straight-line programs
have been developed in \cite{LohreyMR18}.

Finally, recall that our ranked enumeration algorithm outputs the tuples according to an MSO-definable linear order,
whereas in \cite{BourhisGJR21,GMS24} tuples are enumerated in the order of decreasing weights (elements of an 
 ordered abelian group) and the weight
of an input string is computed by a so-called cost transducers.
It is not clear, how these two concepts relate to each other. Under the cost transducer model there can be tuples
with the same weight, so their order in the enumeration is arbitrary. This cannot occur in our model. 


\def\cprime{$'$} \def\cprime{$'$} \def\cprime{$'$} \def\cprime{$'$}

\appendix 

\section{Appendix}

\begin{algorithm}[h]
\SetKwComment{Comment}{(}{)}
\KwIn{edge $e = (u,v)$ of $\calB$, factorization $\pi_e = \pi_1 \pi_2$ with $\pi_1 \neq \varepsilon \neq \pi_2$, $\pi_3 \in \Sat$}
introduce new node $x$ and replace the edge $e$ by two edges $e_1 = (u,x)$, $e_2 = (x,v)$\;
introduce a new leaf node $y$ and the edge $e_3 = (x,y)$\;
$\pi_{e_1} := \pi_1$;  $\pi_{e_2} := \pi_2$; $\pi_{e_3} := \pi_3$\;
 \Return$y$\label{return-split}\;
\caption{\textsf{split}$(e,\pi_1,\pi_2, \pi_3)$  \label{split}}
\end{algorithm}
\noindent
In the algorithm \textsf{right} below, 
whenever \textsf{split} is called within \textsf{right}, the leaf returned in line~\ref{return-split} of \textsf{split} is also returned by \textsf{right}.
In particular, \textsf{right} terminates after every call of \textsf{split}.
\begin{algorithm}[h]
\KwIn{a leaf $u$ in $\calB$}
\KwOut{$\bot$ if $\pos_A(u) = |A|_{\calG}$, else a leaf $u'$ in the new tree with $\pos_A(u') = \pos_A(u)+1$} 
let $e = (v,w)$ be the first edge on the path from the leaf $u$ to the root with $\pi_{e} \notin \sfR$\;
\If{such an edge $e$ does not exist}{\Return $\bot$}
let $\pi_{e} = \pi' (B,\ell,\alpha) \tau$ with $\pi' \in (\sfL \cup \sfR)^*$, $\tau \in \sfR \cup \{\varepsilon\}$; \label{line5} \tcp*[f]{$\pi_{e}$ must have this form} \\
$(B,\ell,C) := \lreduce(B,\ell,\alpha)$  \tcp*[r]{here we set $C := B$ if $\lreduce(B,\ell,\alpha) = \varepsilon$} 
\eIf{$C \neq B$ or $\pi' \neq \varepsilon$\label{line7}} 
  {let $\beta \in V \cup \Sigma$ be the right symbol in $\rho(C)$\;
    \textsf{split}$(e, \; \pi' (B, \ell, C), \; (C, \ell, \alpha) \tau, \; (C,r, \beta)(\beta,\ell,\omega_L(\beta)))$\label{line9}\tcp*[f]{$(\beta,\ell,\omega_L(\beta)) = \varepsilon$ if $\beta \in \Sigma$}}
   (\tcp*[f]{$\alpha$ is the left symbol in $\rho(B)$ and $\pi' = \varepsilon$}){let $e'= (v, x) \neq e$ be the second edge leaving $v$; \tcp*[f]{this edge $e'$ must exist} \\
      let $\pi_{e'} = (B,r, \beta) \pi''$ with $\pi'' \in (\sfL \cup \sfR)^*$;  \label{line-pi''} \tcp*[f]{$\pi_{e'}$ must have this form} \\
      let $\gamma$ be the right symbol in $\rho(B)$\;
      \eIf(\tcp*[f]{we must have $\gamma \in V$}){$\beta \neq \gamma$\label{line14}}
        {\textsf{split}$(e', (B,r,\gamma), (\gamma,r,\beta) \pi'', (\gamma, \ell, \omega_L(\gamma)))$\label{line15}}
        {\eIf{$\pi''$ has the form $(\gamma, \ell,D) \tilde{\pi}$ with $\tilde{\pi} \neq \varepsilon$\label{line17}}
           {\textsf{split}$(e', \; (B,r,\gamma)(\gamma,\ell,D), \; \tilde{\pi}, \; (D,\ell,\omega_L(D)))$\label{line18}}
           (\tcp*[f]{$\pi'' \in \sfL \cup \{\varepsilon\}$ must hold})
           {\eIf{there is an edge $e'' = (y,z)$ such that there is a path of $\sfL$-labelled edges 
              from $x$ to $y$ and $\pi_{e''} \in \sfL \sfR (\sfL \cup \sfR)^*$\label{line23}}
              {let $\pi_{e''} = (E,\ell,F)\tilde{\pi}$ with $\tilde{\pi} \neq \varepsilon$\; 
               \textsf{split}$(e'', \; (E,\ell,F), \; \tilde{\pi}, \; (F,\ell,\omega_L(F)))$\label{line24}}
              {let $y$ be the unique  leaf below $x$ such that every edge on the path from $x$ to $y$ is $\sfL$-labelled;\label{line20}
              \tcp*[f]{$y=x$ is possible} \\
              \Return $y$\label{line21}}}}}
\caption{\textsf{right}$(u)$  \label{right}}
\end{algorithm}
\FloatBarrier

\noindent
Figures~\ref{fig right 1}--\ref{fig right 6} show the different cases arising in Algorithm~\ref{right}. In each figure, the left part shows the situation
before the call of \textsf{right}, whereas the right part shows the situation after the call of \textsf{right}.
Red node labels ($u$, $v$, $w$, $x$, $y$, $z$) refer to the corresponding node names used in Algorithm~\ref{right}.
Edges labelled by $\sfL^*$ (resp., $\sfR^*$) stand for paths that only contain $\sfL$-labelled (resp., $\sfR$-labelled) edges.

\begin{figure}[h]
 \tikzset{alpha/.style={inner sep = 1pt, fill=white}}
 \tikzset{round/.style={inner sep = 1pt, circle, fill=black}}
  \tikzset{empty/.style={inner sep = 0pt}}
\centering{
\begin{tikzpicture}
\draw  (0,0) node[alpha]{\textcolor{red}{$v$}} -- ++(-.4,-.4) --  ++(.4,-.4) -- ++(-.4,-.4) node[alpha,left=.5mm] {$\pi'$} -- ++(.4,-.4) 
   -- ++(-.4,-.4) -- ++(.4,-.4) node[alpha]  {$B$}  -- ++(-2,-2) node[alpha, label= {[label distance=-1.2mm]left:$\textcolor{red}{w}$}] {$\alpha$}  -- ++(.7,-.7) node[pos=.3,below=.5mm] {$\sfR^*$} node[alpha] {\textcolor{red}{$u$}} ;
   
\draw  (4,0) node[alpha]{\textcolor{red}{$v$}} -- ++(-.4,-.4) --  ++(.4,-.4) -- ++(-.4,-.4) node[alpha,left=.5mm] {$\pi'$} -- ++(.4,-.4) 
   -- ++(-.4,-.4) -- ++(.4,-.4) node[alpha] (B) {$B$}  -- ++(-1,-1) node[alpha] (C) {$C$}  -- ++(-1,-1) node[alpha, label= {[label distance=-1.2mm]left:$\textcolor{red}{w}$}] {$\alpha$} -- ++(.7,-.7) node[pos=.3,below=.5mm] {$\sfR^*$} node[alpha] {\textcolor{red}{$u$}} ;
\draw (C) -- ++(1,-1) node[alpha] {$\beta$}  -- ++(-.7,-.7) node[alpha] {\scriptsize{$\omega_L(\beta)$}} ;
\end{tikzpicture}}
\caption{The case $B \neq C$; i.e., $\alpha$ is not the left symbol in $\rho(B)$; see lines \ref{line7}--\ref{line9}.
The path $\pi'$ from $v$ to $B$ can be also empty. We assume moreover that the $\tau$ from line~\ref{line5} is empty (and similarly in the other figures).
There is a similar picture for $\tau \in \sfR$.}  \label{fig right 1}
\end{figure}
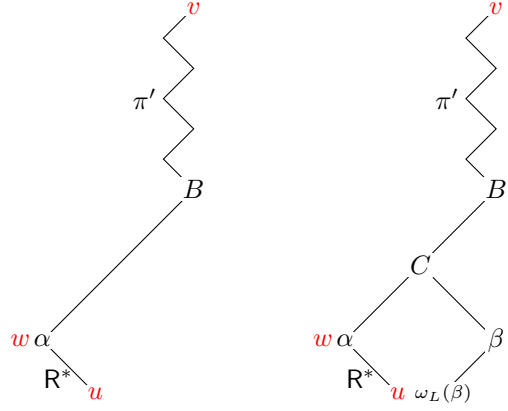

\begin{figure}[h]
 \tikzset{alpha/.style={inner sep = 1pt, fill=white}}
 \tikzset{round/.style={inner sep = 1pt, circle, fill=black}}
  \tikzset{empty/.style={inner sep = 0pt}}
\centering{
\begin{tikzpicture}
\draw  (0,0) node[alpha]{\textcolor{red}{$v$}} -- ++(-.4,-.4) --  ++(.4,-.4) -- ++(-.4,-.4) node[alpha,left=.5mm] {$\pi'$} -- ++(.4,-.4) 
   -- ++(-.4,-.4) -- ++(.4,-.4) node[alpha] {$B$}  -- ++(-1,-1) node[alpha, label= {[label distance=-1.2mm]left:$\textcolor{red}{w}$}] {$\alpha$}  -- ++(.7,-.7) node[pos=.3,below=.5mm] {$\sfR^*$} node[alpha] {\textcolor{red}{$u$}} ;
   
\draw  (4,0) node[alpha]{\textcolor{red}{$v$}} -- ++(-.4,-.4) --  ++(.4,-.4) -- ++(-.4,-.4) node[alpha,left=.5mm] {$\pi'$} -- ++(.4,-.4) 
   -- ++(-.4,-.4) -- ++(.4,-.4) node[alpha] (B) {$B$}  -- ++(-1,-1) node[alpha, label= {[label distance=-1.2mm]left:$\textcolor{red}{w}$}] {$\alpha$} -- ++(.7,-.7) node[pos=.3,below=.5mm] {$\sfR^*$} node[alpha] {\textcolor{red}{$u$}} ;
\draw (B) -- ++(1,-1) node[alpha]{$\beta$}  -- ++(-.7,-.7) node[alpha] {\scriptsize{$\omega_L(\beta)$}} ;
\end{tikzpicture}}
\caption{The case $B = C$ (i.e., $\alpha$ is the left symbol in $\rho(B)$) and $\pi' \neq \varepsilon$; see lines \ref{line7}--\ref{line9}.}  \label{fig right 2}
\end{figure}
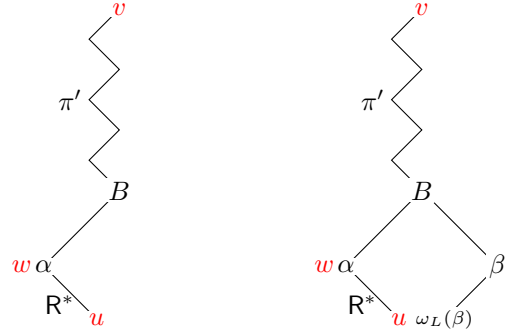

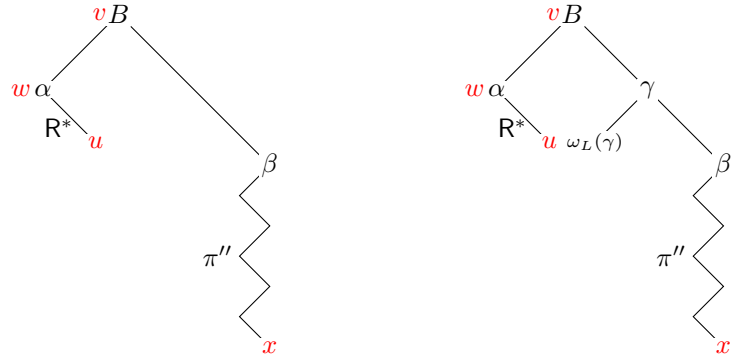
\begin{figure}[h]
 \tikzset{alpha/.style={inner sep = 1pt, fill=white}}
 \tikzset{round/.style={inner sep = 1pt, circle, fill=black}}
  \tikzset{empty/.style={inner sep = 0pt}}
\centering{
\begin{tikzpicture}
\draw  (0,0) node[alpha, label= {[label distance=-1.5mm]left:$\textcolor{red}{v}$}] (B) {$B$}  -- ++(-1,-1) node[alpha, label= {[label distance=-1.2mm]left:$\textcolor{red}{w}$}] {$\alpha$}  -- ++(.7,-.7) node[pos=.3,below=.5mm] {$\sfR^*$} node[alpha] {\textcolor{red}{$u$}} ;
\draw (B) -- ++(2,-2) node[alpha]{$\beta$} -- ++(-.4,-.4) --  ++(.4,-.4) -- ++(-.4,-.4)  node[alpha,left=.5mm] {$\pi''$}  -- ++(.4,-.4) 
   -- ++(-.4,-.4) -- ++(.4,-.4)  node[alpha] {\textcolor{red}{$x$}} ;
   
\draw  (6,0) node[alpha, label= {[label distance=-1.5mm]left:$\textcolor{red}{v}$}] (B) {$B$}  -- ++(-1,-1) node[alpha, label= {[label distance=-1.2mm]left:$\textcolor{red}{w}$}] {$\alpha$} -- ++(.7,-.7) node[pos=.3,below=.5mm] {$\sfR^*$} node[alpha] {\textcolor{red}{$u$}} ;
\draw (B) -- ++(1,-1) node[alpha] (gamma) {$\gamma$} -- ++(1,-1) node[alpha]{$\beta$} -- ++(-.4,-.4) --  ++(.4,-.4) -- ++(-.4,-.4)  node[alpha,left=.5mm] {$\pi''$}  -- ++(.4,-.4) 
   -- ++(-.4,-.4) -- ++(.4,-.4) node[alpha] {\textcolor{red}{$x$}} ;
\draw (gamma) -- ++(-.7,-.7) node[alpha] {\scriptsize{$\omega_L(\gamma)$}} ;
\end{tikzpicture}}
\caption{Here, $\rho(B) = \alpha\gamma$ and $\beta \neq \gamma$; see lines \ref{line14} and \ref{line15}.}  \label{fig right 3}
\end{figure}

\begin{figure}[t]
 \tikzset{alpha/.style={inner sep = 1pt, fill=white}}
 \tikzset{round/.style={inner sep = 1pt, circle, fill=black}}
  \tikzset{empty/.style={inner sep = 0pt}}
\centering{
\begin{tikzpicture}
\draw  (0,0) node[alpha, label= {[label distance=-1.5mm]left:$\textcolor{red}{v}$}] (B) {$B$}  -- ++(-1,-1) node[alpha, label= {[label distance=-1.2mm]left:$\textcolor{red}{w}$}] {$\alpha$}  -- ++(.7,-.7) node[pos=.3,below=.5mm] {$\sfR^*$} node[alpha] {\textcolor{red}{$u$}} ;
\draw (B) -- ++(1,-1) node[alpha]{$\gamma$} -- ++(-.7,-.7) node[alpha] {$D$}  --  ++(.7,-.7) -- ++(-.4,-.4) -- ++(.4,-.4) node[alpha,right=.5mm] {$\tilde{\pi}$} 
   -- ++(-.4,-.4) -- ++(.4,-.4) -- ++(-.4,-.4) node[alpha] {\textcolor{red}{$x$}} ;
   
\draw  (6,0) node[alpha, label= {[label distance=-1.5mm]left:$\textcolor{red}{v}$}] (B) {$B$}  -- ++(-1,-1) node[alpha, label= {[label distance=-1.2mm]left:$\textcolor{red}{w}$}] {$\alpha$}  -- ++(.7,-.7) node[pos=.3,below=.5mm] {$\sfR^*$} node[alpha] {\textcolor{red}{$u$}} ;
\draw (B) -- ++(1,-1) node[alpha]{$\gamma$} -- ++(-.7,-.7) node[alpha] (D) {$D$}  --  ++(.7,-.7) -- ++(-.4,-.4) -- ++(.4,-.4) node[alpha,right=.5mm] {$\tilde{\pi}$} 
   -- ++(-.4,-.4) -- ++(.4,-.4) -- ++(-.4,-.4) node[alpha] {\textcolor{red}{$x$}} ;
\draw (D) -- ++(-.7,-.7) node[alpha] {\scriptsize $\omega_L(D)$} ;
\end{tikzpicture}}
\caption{Here, $\rho(B) = \alpha \gamma$ and $\pi''$ from line~\ref{line-pi''}  has the form $(\gamma, \ell,D) \tilde{\pi}$ with $\tilde{\pi} \neq \varepsilon$; see lines \ref{line17} and \ref{line18}.}  \label{fig right 4}
\end{figure}
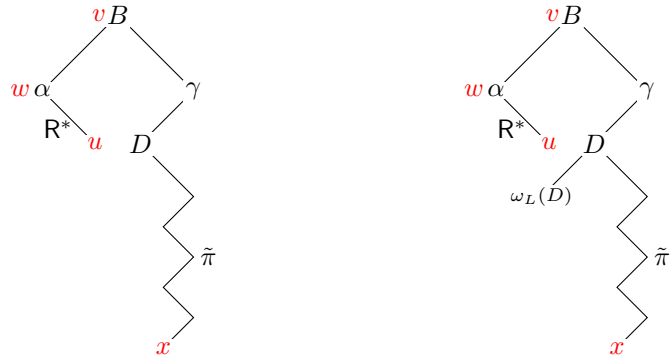

\begin{figure}[t]
 \tikzset{alpha/.style={inner sep = 1pt, fill=white}}
 \tikzset{round/.style={inner sep = 1pt, circle, fill=black}}
  \tikzset{empty/.style={inner sep = 0pt}}
\centering{
\begin{tikzpicture}
\draw  (0,0) node[alpha, label= {[label distance=-1.5mm]left:$\textcolor{red}{v}$}] (B) {$B$}  -- ++(-1.5,-1.5) node[alpha, label= {[label distance=-1.2mm]left:$\textcolor{red}{w}$}] {$\alpha$}  -- ++(1,-1) node[pos=.3,below=.5mm] {$\sfR^*$} node[alpha] {\textcolor{red}{$u$}} ;
\draw (B) -- ++(1.5,-1.5) node[alpha, label= {[label distance=-1.2mm]right:$\textcolor{red}{x}$}]{$\gamma$} -- ++(-1,-1)  node[pos=.3,below=.5mm] {$\sfL^*$} node[alpha, label= {[label distance=-1.2mm]right:$\textcolor{red}{y}$}] (E) {$E$}  -- ++(-.5,-.5) node[alpha] {$F$}  
-- ++(.4,-.4) -- ++(-.4,-.4) -- ++(.4,-.4) node[alpha,right=.5mm] {$\tilde{\pi}$}  -- ++(-.4,-.4) -- ++(.4,-.4) - ++(-.4,-.4) node[alpha] {\textcolor{red}{$z$}} ; 
\draw (E) -- ++(.5,-.5) ;

\draw  (6,0) node[alpha, label= {[label distance=-1.5mm]left:$\textcolor{red}{v}$}] (B) {$B$}  -- ++(-1.5,-1.5) node[alpha, label= {[label distance=-1.2mm]left:$\textcolor{red}{w}$}] {$\alpha$}  -- ++(1,-1) node[pos=.3,below=.5mm] {$\sfR^*$} node[alpha] {\textcolor{red}{$u$}} ;
\draw (B) -- ++(1.5,-1.5) node[alpha, label= {[label distance=-1.2mm]right:$\textcolor{red}{x}$}]{$\gamma$} -- ++(-1,-1)  node[pos=.3,below=.5mm] {$\sfL^*$} node[alpha, label= {[label distance=-1.2mm]right:$\textcolor{red}{y}$}] (E) {$E$}  -- ++(-.5,-.5) node[alpha] (F) {$F$}  
-- ++(.4,-.4) -- ++(-.4,-.4) -- ++(.4,-.4) node[alpha,right=.5mm] {$\tilde{\pi}$}  -- ++(-.4,-.4) -- ++(.4,-.4) - ++(-.4,-.4) node[alpha] {\textcolor{red}{$z$}} ; 
\draw (E) -- ++(.5,-.5) ;
\draw (F) -- ++(-.7,-.7) node[alpha] {\scriptsize $\omega_L(F)$} ;
\end{tikzpicture}}
\caption{Here, $\rho(B) = \alpha \gamma$ and there is an edge $e'' = (y,z)$ such that there is a path of $\sfL$-labelled edges from $x$ to $y$ and $\pi_{e''} = (E,\ell,F)\tilde{\pi}$ with $\tilde{\pi} \neq \varepsilon$; see lines \ref{line23}--\ref{line24}. We assume moreover that $\pi''$ from line~\ref{line-pi''}  is empty.
There is a similar picture for $\pi'' \in \sfL$.}  \label{fig right 6}
\end{figure}
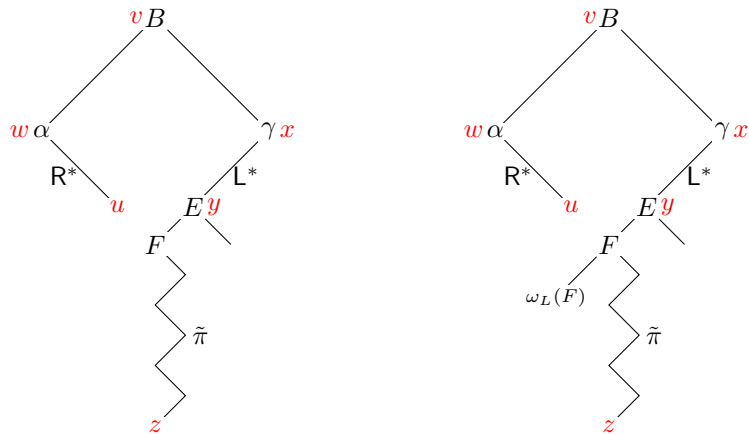

\begin{figure}[t]
 \tikzset{alpha/.style={inner sep = 1pt, fill=white}}
 \tikzset{round/.style={inner sep = 1pt, circle, fill=black}}
  \tikzset{empty/.style={inner sep = 0pt}}
\centering{
\begin{tikzpicture}
\draw  (0,0) node[alpha, label= {[label distance=-1.5mm]left:$\textcolor{red}{v}$}]
(B) {$B$}  -- ++(-1.5,-1.5) node[alpha, label= {[label distance=-1.2mm]left:$\textcolor{red}{w}$}] {$\alpha$}  -- ++(1,-1) 
node[pos=.3,below=.5mm] {$\sfR^*$} node[alpha] {\textcolor{red}{$u$}} ;
\draw (B) -- ++(1.5,-1.5) node[alpha, label= {[label distance=-1.2mm]right:$\textcolor{red}{x}$}]{$\gamma$} -- ++(-1,-1)  node[pos=.3,below=.5mm] {$\sfL^*$} node[alpha] {\textcolor{red}{$y$}}; 
\end{tikzpicture}}
\caption{Here, $\rho(B) = \alpha \gamma$ and  
$y$ is a leaf; see lines \ref{line20} and \ref{line21}. In this situation, the tree $\calB$ is not modified.}  \label{fig right 5}
\end{figure}
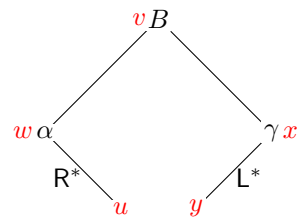

\end{document}